\pdfoutput=1
\documentclass{article}
\usepackage{fullpage}
\usepackage[T1]{fontenc}
\usepackage{mathpazo}
\usepackage{graphicx}
\usepackage{amssymb}
\usepackage{amsmath,amsfonts,amsthm}
\usepackage{mathtools,thmtools}

\usepackage[ruled,linesnumbered]{algorithm2e}
\usepackage{natbib}
\usepackage{tikz}
\usepackage{nicefrac}
\usepackage{wrapstuff}
\usepackage{booktabs}
\usepackage{float}
\usepackage[dvipsnames]{xcolor}
\usepackage{tcolorbox}
\usepackage{subcaption}
\usepackage{enumitem}

\usepackage[colorlinks=true, allcolors=blue]{hyperref}
\usepackage[nameinlink]{cleveref}
\usepackage{crossreftools} \AddToHook{cmd/appendix/before}{\crefalias{section}{appendix}\crefalias{subsection}{appendix}
}

\newtheorem{theorem}{Theorem}[section]
 \newtheorem{lemma}[theorem]{Lemma}
\newtheorem{corollary}[theorem]{Corollary}
\newtheorem{proposition}[theorem]{Proposition}
\theoremstyle{definition}
\newtheorem{definition}[theorem]{Definition}
\newtheorem{example}[theorem]{Example}
\theoremstyle{remark}

\setlist[enumerate]{label=(\arabic*), ref=\arabic*}

\newcommand{\showeprint}[2][]{\href{https://arxiv.org/abs/#2}{arXiv:\nolinkurl{#2}}}

\renewcommand{\ge}{\geqslant}
\renewcommand{\geq}{\geqslant}
\renewcommand{\le}{\leqslant}
\renewcommand{\leq}{\leqslant}
\renewcommand{\epsilon}{\varepsilon}
\DeclarePairedDelimiter{\set}{\lbrace}{\rbrace}

\DeclareMathOperator{\argmin}{\arg\min}
\DeclareMathOperator{\argmax}{\arg\max}
\newcommand{\N}{\mathbb{N}}
\newcommand{\R}{\mathbb{R}}

\newcommand{\myparagraph}[1]{\medskip\noindent\textbf{#1}~}
\newcommand{\proofstep}[1]{\medskip\noindent\emph{#1}~}

\newcommand{\calU}{\mathcal{U}}
\newcommand{\calC}{\mathcal{C}}
\newcommand{\calCup}{\calC_{\ge}}
\newcommand{\calR}{\mathcal{R}}
\newcommand{\calX}{\mathcal{X}}
\newcommand{\calY}{\mathcal{Y}}
\newcommand{\calI}{\mathcal{I}}
\newcommand{\calA}{\mathcal{A}}
\newcommand{\PE}{\textrm{eff}}

\newcommand{\conv}{\operatorname{conv}}

\newcommand{\supp}{\operatorname{supp}}
\newcommand{\NW}{\operatorname{NW}}
\newcommand{\Nash}{\operatorname{Nash}}
\newcommand{\Prod}{\operatorname{NP}}
\makeatletter
\newcommand{\tableofcontentswithouttitle}{\@starttoc{toc}}
\makeatother

\title{Proportional Fairness for Harmful Decisions}
\author{
    Benjamin Cookson\textsuperscript{1} \and 
    Soroush Ebadian\textsuperscript{1} \and 
    Dominik Peters\textsuperscript{2} \and 
    Nisarg Shah\textsuperscript{1}
}
\date{
    \small
    \textsuperscript{1}University of Toronto \\
    \textsuperscript{2}CNRS, LAMSADE, Université Paris Dauphine - PSL \\
    \{bcookson,soroush,nisarg\}@cs.toronto.edu, dominik.peters@lamsade.dauphine.fr
    \vspace{-5pt}
}

\begin{document}

\maketitle

\begin{abstract}
We study allocation of (divisible) public bads, where agents incur costs for alternatives and the goal is to pick a lottery over the alternatives. We show that the traditional definitions of the core, a central criterion of proportional representation for allocation of public goods, private goods, and private bads (chores), do not make sense for allocation of public bads. 

We introduce two formalizations of the core tailored to public bads. Under a structural condition which subsumes allocation of private bads, we show that zero-respecting Lindahl equilibria satisfy both formalizations, exhibit additional fairness guarantees, and strictly generalize competitive equilibria from equal incomes (CEEI) for allocation of private bads. Without this structural condition, we show that Lindahl equilibria exhibit undesirable behaviors, prove sharp impossibility results separating public bads from public goods, but show that a rule using a reduction to public goods recovers one of our formalizations of the core. Our results lay the groundwork for studying fair allocation of public bads, an overlooked yet fundamental problem, and highlight several structural and algorithmic directions that remain open.
\end{abstract}

\vspace{16pt}
\hrule 
\setcounter{tocdepth}{2} \tableofcontentswithouttitle
\medskip
\hrule 

\section{Introduction}\label{sec:intro}
Many collective decisions concern not the distribution of benefits, but the allocation of burdens. Communities must decide where to put landfills, power plants, or highways, which imposes noise, pollution, or disruption on nearby residents. Regulators must determine how compliance costs for new environmental targets are distributed across sectors. During periods of scarcity, utilities must choose which neighborhoods experience load shedding and at what times. In each case, society faces a \emph{public} decision: a single outcome must be chosen, and that outcome imposes (potentially different) costs on each affected party.

Despite the ubiquity of such problems, the theory of \emph{public bads allocation} remains strikingly underdeveloped compared to its well-studied counterparts of public goods, private goods, and private bads (known as chores). The allocation of \emph{private goods} (divisible resources distributed among agents) and \emph{public goods} (a single outcome chosen for all, where agents derive positive utility) are classical subjects of study in economics and social choice, with well-understood solution concepts including competitive equilibrium from equal incomes (CEEI) / Lindahl equilibrium, proportional fairness, and maximum Nash welfare. These concepts turn out to be equivalent under linear utilities and satisfy compelling fairness properties such as the core. More recently, the allocation of \emph{private bads} has received significant attention, with \citet{BMSY17} providing an elegant characterization of CEEI in terms of critical points of the Nash product on the Pareto frontier. Yet public bads---the fourth and final cell of this two-by-two classification---has been largely neglected.

A natural instinct is to reduce public bads to public goods by simply negating costs to obtain utilities. However, the lessons learned from the transition from \textit{private} goods to \textit{private} bads suggest that moving from utilities to costs fundamentally changes the structure of fair allocation. Indeed, the equivalence between market equilibria, proportional fairness, and maximum Nash welfare---which holds for both public and private goods---breaks down for private bads: CEEI allocations correspond not to Nash welfare maximizers but to \emph{critical points} of the Nash product on the efficient frontier, and multiple equilibria with distinct cost vectors may coexist.

A similar breakdown is to be expected for public bads. But while, for private bads, market equilibria remain well-defined and satisfy the core, the public case introduces yet another layer of difficulty. The issue is that fairness, and in particular the notion of the \emph{core}, becomes hard to define and is conceptually problematic. For private goods and bads, the core asks whether any coalition $S$ could ``go it alone'' and divide its proportional share of the resources among its members to make everyone in $S$ better off. For public goods, this extends naturally: a coalition considers alternative outcomes, with its share of influence scaling with its size. But for public bads, neither of these approaches makes sense, because any outcome chosen by a coalition will impose \emph{externalities} on non-members. There is no obvious way to ``allocate all the bads among coalition members'' for public outcomes. This conceptual gap leaves us without a formal notion of fairness for public bads.

These observations motivate the central questions of our work:

\begin{quote}
    \emph{What is the right definition of Lindahl equilibrium for public bads, and how does it relate to proportional fairness? Can one define a meaningful notion of the core when coalitional deviations inevitably impose externalities on outsiders? What fairness guarantees and axioms are achievable, and what are the fundamental trade-offs among them?}
\end{quote}

\subsection{Our Results \& Techniques}\label{sec:results}

Our main contribution is to provide a theory of public bads allocation, paralleling the established frameworks for public goods and private bads allocation where relevant, while also highlighting where the techniques from those areas break down, and new ideas must be developed.

We begin the paper in \Cref{sec:public-goods} with a systematic review of what is known about the fair allocation of public and private goods, explaining the connections between Lindahl equilibrium, maximum Nash welfare, proportional fairness, and CEEI.

In \Cref{sec:public-bads}, we discuss difficulties that arise when extending the theory of allocating public goods to the allocation of public bads, as well as when generalizing results about the allocation of private bads to public bads. We introduce several fairness notions designed specifically for public bads allocation problems, including two, named \emph{Bounded Externality Core} and \emph{Completion Core}, which aim to capture the spirit of core-based fairness properties that are discussed in the literatures on public goods and private bads.

In \Cref{sec:bads-lindahl}, we provide a formal definition of a Lindahl equilibrium for public bads, building upon a recent proposal of \citet{TZ25}. We show that the set of \emph{zero-respecting} Lindahl allocations (Lindahl allocations where every agent receives a strictly positive cost) have attractive properties: they correspond exactly to the well-studied concept of CEEIs when restricted to instances corresponding to private bads allocation, and they achieve very strong fairness guarantees when restricted to what we call \emph{axes-cutting} instances. However, we show that outside of these special cases, zero-respecting Lindahl equilibria lose much of their desirability, and can exhibit notably unfair behavior. This is a surprising result, since Lindahl equilibria are always fair in the goods world.

In \Cref{sec:impossibilities}, we aim to find fair solutions beyond the special cases in \Cref{sec:bads-lindahl}. We give an axiomatic characterization of zero-respecting Lindahl allocations by generalizing a result of \citet{mariotti2005nash}. This characterization highlights the properties which cause the unfair behavior of Lindahl on general instances. Based on this, we introduce a new rule called \emph{Flipped-MNW} which does achieve favorable fairness properties on general instances, though it fails to behave as nicely as zero-respecting Lindahl allocations on the special cases we studied in \Cref{sec:bads-lindahl}.

Finally, in \Cref{sec:computation}, we address algorithmic considerations. We discuss a method for enumerating all essentially different Lindahl allocations for a given instance in an approximate sense. We also give a (non-convex) optimization program whose KKT points correspond to zero-respecting Lindahl equilibria, and show that in the axes-cutting special case, the greedy Frank--Wolfe algorithm proposed by \citet{chaudhury2024competitive} can be applied to this program to efficiently find a zero-respecting Lindahl equilibrium.

\subsection{Related Work}\label{sec:related}

\myparagraph{Public goods allocation.}
The problem of fairly dividing a budget between several public goods has been studied in several models \citep{BMS2005,aziz2019fair,airiau2023portioning,brandl2021distribution}. \citet{FGM16} connected this problem to the Lindahl equilibrium introduced by \citet{Fol70}, which was further studied by \citet{munagala2022coremultilinear}, \citet{KP25}, and \citet{TZ25}. Indivisible versions of this problem have also received a lot of attention \citep{FMS18}, including as part of the literature on committee elections \citep{lackner2023multi} and participatory budgeting \citep{AS21,rey2023pb}.

\myparagraph{Fair division of private goods and chores.}
The allocation of divisible private goods via competitive equilibrium from equal incomes (CEEI) is classical~\citep{Var74}. For linear utilities, this is equivalent to maximizing Nash welfare~\citep{eisenberg1961aggregation}, which allows applying convex programming and other techniques to efficiently find CEEI allocations \citep{Vazi07a}.
For private bads (chores), \citet{BMSY17} characterize CEEI allocations as critical points of the Nash product on the Pareto frontier. They also discuss axiomatic properties of the CEEI solution. In particular, they establish that CEEI satisfies the core and always yields strictly positive costs---properties that turn out to be challenging to port to public bads.
Several algorithms for computing CEEI for chores have been proposed \citep{branzei2024algorithms,boodaghians2022polynomial,chaudhury2024competitive}.
The allocation of goods and chores has been intensely studied for \textit{indivisible} items. In particular, the interaction of fairness (envy-freeness up to one item) and efficiency (Pareto optimality) has been a major theme: for goods, the maximum Nash welfare rule achieves both properties~\citep{CKMP+19}; the recent breakthrough of \citet{mahara2025existence} shows that the two properties are also compatible for chores.
We refer to the surveys of \citet{AABFLMVW23} and \citet{liu2024mixed} for broader coverage of fair division.

\myparagraph{Proportional fairness in social choice.}
The proportional fairness (PF) criterion we study originates in the work of \citet{Kelly97} in the context of network resource allocation. It has been applied to (probabilistic) social choice by \citet{EKPS22} in a distortion framework and by \citet{BGHJ+23} in an online setting. The concept was also recently applied to facility location and clustering~\citep{CFLM19,MS20,CMS24,KP24}.
The PF criterion implies notions of the core, and the core has been used as an axiomatic formalization of \textit{proportional representation} in many subfields of computational social choice in recent years. Prominent examples include participatory budgeting~\citep{FMS18,aziz2019fair,brandl2022funding,airiau2023portioning} and committee selection~\citep{ABCE+17,lackner2023multi,brill2023proportionality,peters2025core}.

\paragraph{Previous work on public bads selection.}
The work most closely related to ours is that of \citet{mariotti2005nash}. They study the problem of allocating losses in utility space relative to a fixed reference vector by selecting a utility vector from a feasible set, and propose the \emph{Nash rationing solution} for this problem. This solution is closely related to our strictly positive PF allocations. Specifically, our public bads model can be translated into their setting by taking the reference vector to be $0$ and the feasible utility set to be the downward closure of the negations of the feasible cost vectors induced by all lotteries. The key difference between our work and that of \citet{mariotti2005nash} is that we focus on defining fairness notions for public bads and studying the relationships among PF, Lindahl equilibrium, and competitive equilibrium, whereas \citet{mariotti2005nash} focus primarily on an axiomatic analysis and existence of the Nash rationing solution and other loss allocation rules. We leverage several of their results, and discuss similarities in greater detail in \Cref{sec:axes-cutting,sec:impossibilities}.

\section{Review: Allocation of Public (and Private) Goods}\label{sec:public-goods}

\myparagraph{Preliminaries.} For $t \in \N$, define $[t]\triangleq\set{1,2,\ldots,t}$. For a finite set $S$, $\Delta(S)\triangleq \set{x \in \R_{\ge 0}^S:\sum_{a\in S} x(a) =1}$ is the probability simplex over $S$. For a lottery $x \in \Delta(S)$, its support is $\supp(x)=\set{a\in S : x(a)>0}$. Given $x \in \R^N$ and $S \subseteq N$, define $x$ restricted to $S$ as $x_S \triangleq (x_i)_{i \in S}$. We sometimes view a mapping $p : S \to \R$ as a vector $p \in \R^S$. The dot product of vectors $p,q \in \R^S$ is $p \cdot q \triangleq \sum_{s \in S} p(s) \cdot q(s)$. For vectors $p,q \in \R^k$, we say $p \leq q$ (equivalently, $q \geq p$) when $p_i \le q_i$ for all $i \in [k]$, and $p < q$ (equivalently, $q > p$) if $p \leq q$ and $p \neq q$. We write $p \gg q$ when $p_i > q_i$ for all $i \in [k]$.

\myparagraph{Model of public goods.} A public goods instance is given by the tuple $(N,A,v)$, where $N = [n]$ is a set of agents, $A$ is a set of $m$ alternatives, and $v = (v_1,\ldots,v_n)$, where $v_i : A \to \R_{\ge 0}$ is the (von Neumann--Morgenstern) \emph{valuation function} of agent $i \in N$. Assume that for each agent $i \in N$, $v_i(a) > 0$ for at least one alternative $a \in A$ (otherwise the agent always gets zero utility and can be ignored). These valuations induce linear utilities over lotteries: for $x \in \Delta(A)$, define, with slight abuse of notation, $v_i(x) \triangleq \sum_{a\in A} x(a) \cdot v_i(a)$, and $v(x) \triangleq (v_1(x),\ldots,v_n(x))$ to be the utility vector induced by $x$. The (feasible) \emph{utility set} is
\[
\calU\triangleq\set{v(x) : x\in \Delta(A)} = \conv\set{v(a) : a\in A} \subseteq\R_{\ge 0}^N,
\]
where the equality is due to linearity of utilities. $\calU$ is a non-empty, compact, and convex polytope. 

For this model, three well-known rules coincide and yield a particularly appealing solution concept. The first is Lindahl equilibrium, a virtual market equilibrium defined by \citet{Fol70} drawing on ideas of \citet{Lind58}, in which every agent simultaneously finds the same lottery $x$ to be utility-maximizing subject to their personalized prices and a common budget. 

\begin{definition}[Lindahl Equilibrium for Public Goods]
For public goods, given a lottery $x \in \Delta(A)$ and personalized prices $p = (p_i : A \to \R_{\ge 0})_{i \in N}$, we say that $(x,p)$ is a \emph{Lindahl equilibrium} if:
\begin{enumerate}
\item \emph{Spending constraint:} $p_i \cdot x \leq \nicefrac{1}{n}$ for all $i \in N$.
\item \emph{Utility maximization:} For every $i \in N$ and $y \in \R_{\geq 0}^A$ such that $p_i \cdot y \leq \nicefrac{1}{n}$, we have $v_i(y) \leq v_i(x)$.
\item \emph{Profit maximization:} $\sum_{i \in N} p_i(a) \leq 1$ for every $a \in A$, with equality if $x(a) > 0$.
\end{enumerate}
Finally, $x$ is a \emph{Lindahl allocation} if $(x,p)$ is a Lindahl equilibrium for some personalized prices $p$.
\end{definition}

Note that the spending constraint condition is sometimes referred to as affordability.

\citet{Fol70} proves the existence of such an equilibrium under general convex and strictly monotone utility functions. \citet{FGM16} provide an efficient algorithm via convex programming for a class of utility functions they term ``non-satiating''. For our special case of linear utilities, this convex program simply maximizes $\sum_{i \in N} \log v_i(x)$~\cite[Corollary 2.3]{FGM16}, or equivalently, the Nash welfare $\prod_{i \in N} v_i(x)$; its KKT conditions turn out to be precisely the Lindahl equilibrium conditions. 

\begin{definition}[Maximum Nash Welfare (MNW)]
The \emph{Nash welfare} of a lottery $x\in \Delta(A)$ is $\NW(x) \triangleq \prod_{i\in N} v_i(x)$. We say that $x$ is a \emph{maximum Nash welfare} (MNW) lottery if $x \in \argmax_{y \in \Delta(A)} \NW(y)$.
\end{definition}

First-order optimality for this convex program yields the following alternative formalization, dating back to the work of \citet{Kelly97}, which has been explored in social choice~\citep{EKPS22,BGHJ+23} and beyond~\citep{RLKL+22}. For a formal proof of equivalence, see the works of \citet[Proof of Theorem 2]{EFS24} and \citet[Proof of Theorem 4]{KP25}.

\begin{definition}[PF for Goods]\label{def:pf-public-goods}
A lottery $x \in \Delta(A)$ is \emph{proportionally fair} if we have $\frac{1}{n}\sum_{i\in N}\frac{v_i(y)}{v_i(x)}\le 1$ for all $y\in \Delta(A)$, with the convention that $\nicefrac{0}{0}=1$ and $\nicefrac{\alpha}{0} = +\infty$ for any $\alpha > 0$. Due to linear utilities, this is equivalent to having $\frac{1}{n}\sum_{i\in N}\frac{v_i(a)}{v_i(x)}\le 1$ for all $a\in A$ with equality if $x(a) > 0$. 
\end{definition}

\citet{Fol70} proves that Lindahl equilibria satisfy a compelling fairness desideratum.\footnote{Technically, \citet{Fol70} proves that they are in the \emph{weak core} in a more general setup; the fact that they are in the core for our setup can be inferred from Proposition~2 of \citet{KP25}, whose ``cap-sufficiency'' and ``zero-respecting'' conditions automatically hold in our ``uncapped'' setting. See also \citet[Theorem 3]{aziz2019fair}.}

\begin{definition}[The Core for Public Goods]\label{def:core-public-goods}
A lottery $x \in \Delta(A)$ is in the \emph{core} if there exist no $S\subseteq N$, $y\in\Delta(A)$ such that $\frac{|S|}{n} \cdot v_i(y)\ge v_i(x)$ for all $i\in S$, with at least one inequality being strict.
\end{definition}

This is a strengthening of Pareto optimality, which imposes the above condition only for $S = N$. 

\begin{definition}[Pareto Optimality]
A lottery $x \in \Delta(A)$ is \emph{Pareto optimal} if there exists no $y\in\Delta(A)$ such that $v_i(y)\ge v_i(x)$ for all $i\in N$, with at least one inequality being strict.
\end{definition}

We will not reproduce the Pareto optimality definition for the other settings---private goods, public bads, and private bads---as it is essentially identical, with the only difference being a reversed inequality between costs for the case of bads. 

The following folklore result follows from the discussion above; we give a proof in \Cref{app:public-goods} for completeness. 

\begin{restatable}{theorem}{publicGoods}\label{thm:public-goods}
    For any public goods instance with linear utilities, the following are equivalent for a lottery $x \in \Delta(A)$.
    \begin{enumerate}
        \item $x$ is a Lindahl allocation. 
        \item $x$ is proportionally fair (PF).
        \item $x$ is a maximum Nash welfare (MNW) lottery.
    \end{enumerate}
    Further, all such lotteries $x$ lie in the core for public goods (and hence, are Pareto optimal) and induce the same utility vector $v(x)$, which gives a strictly positive utility $v_i(x) > 0$ to every agent $i \in N$.
\end{restatable}

\subsection{Relation to Private Goods}\label{sec:private-goods}

Allocation of public goods strictly generalizes allocation of divisible private goods. A private goods instance is given by the tuple $(N,G,v)$, where $N = [n]$ is a set of agents, $G$ is a set of divisible goods, and $v = (v_1,\dots,v_n)$, where $v_i : G \rightarrow \R_{\geq 0}$ and $v_i(g)$ is the value of agent $i$ for good $g$. These values induce linear utilities, given by $v_i(z) \triangleq \sum_{g \in G} v_{i}(g) \cdot z(g)$ for every (fractional) bundle $z \in [0,1]^G$. We assume that $v_{i}(g) > 0$ for at least one $g \in G$, otherwise the agent can be safely ignored with no goods allocated to them. Similarly, we assume that for each $g \in G$, $v_i(g) > 0$ for at least one $i \in N$.

The goal is to find a fractional allocation $x = (x_1,\dots,x_n)$, with $x_i: G \rightarrow [0,1]$ and $x_{i}(g)$ denoting the fraction of good $g \in G$ allocated to agent $i \in N$, that is feasible: $\sum_{i \in N} x_{i}(g) = 1$ for all goods $g \in G$. The utility of agent $i$ under this allocation is $v_i(x_i)$.

\myparagraph{Reduction $\calR$ from private to public goods.} Any private goods instance $(N,G,v)$ can be reduced to a public goods instance $(N,A,v) \triangleq \calR(N,G,v)$ with the set of alternatives $A$ consisting of all $n^m$ \emph{integral allocations} $x$ satisfying $x_{i}(g) \in \set{0,1}$ for all $i \in N$ and $g \in G$; linear utilities for lotteries over integral allocations precisely map to linear utilities for fractional allocations. Specifically, for any allocation $x$ over a private goods instance $(N,G,v)$, we say that a public allocation $y$ over $\calR(N,G,v)$ is the public allocation \emph{induced} by $x$ if for each $i \in N$ and $g \in G$, we have that $x_i(g) = \sum_{z : y(z) > 0}y(z) \cdot z_i(g)$. We refer to $(N,A,v)$ as a \emph{private-induced public goods instance}, induced by $(N,G,v)$.

\myparagraph{Solution concepts for private goods.} In economics, the private goods model is known as a Fisher market, which admits the following market equilibrium concept.
\begin{definition}[CEEI for Private Goods]
    For private goods, we say that the pair $(x,p)$ of a fractional allocation $x$ and a price vector $p \in \R_{> 0}^G$ is a \emph{competitive equilibrium from equal incomes} (CEEI) if:
    \begin{enumerate}
    \item \emph{Maximum bang-per-buck:}
    For all $i \in N$ and $g \in G$, $x_{i}(g) > 0 \Rightarrow g \in \argmax_{g' \in G} \frac{v_{i}(g')}{p(g')}$.
    \item \emph{Budget exhaustion:} $p \cdot x_i = \nicefrac{1}{n}$ for all $i\in N$.
    \end{enumerate}
    We say that $x$ is a CEEI allocation if $(x,p)$ is a CEEI for some $p$.
\end{definition}

In a private goods instance $(N,G,v)$, the set of CEEI allocations (i.e., allocations $x$ that form a CEEI $(x,p)$ with some price vector $p$) coincides with the sets of maximum Nash welfare allocations~\citep{eisenberg1961aggregation} and proportionally fair allocations~\cite[Volume 2, Chapter 14]{AIHS81}. For the induced public goods instance $\calR(N,G,v)$, these maximum Nash welfare allocations coincide with Lindahl allocations (by \Cref{thm:public-goods}). Hence, Lindahl equilibria generalize CEEI for private goods to public goods.

\citet{Var74} shows that CEEI allocations satisfy the core for private goods.\footnote{\citet{Var74} shows that CEEI allocations satisfy the \emph{weak} core in a model allowing more general utility functions. A folklore proof that they satisfy the core under linear utilities is given by \citet[Appendix B]{FSV20}.}

\begin{definition}[The Core for Private Goods]\label{def:core-private-goods}
An allocation $x$ of private goods is in the \emph{core} if there exist no subset of agents $S\subseteq N$ and allocation $y$ of the goods to the agents in $S$ (i.e., $y_{j}(g) = 0$ for all $j \in N\setminus S$ and $g \in G$) such that $\frac{|S|}{n} \cdot v_i(y) \ge v_i(x)$ for all $i\in S$, and at least one inequality is strict.
\end{definition}

It is easy to see that the core for a private goods instance $(N,G,v)$ coincides with the core for its induced public goods instance $\calR(N,G,v)$: the core definition for $\calR(N,G,v)$ allows a coalition $S$ to pick any lottery over alternatives (integral allocations), but any violation can also be produced by a lottery over alternatives in which all the private goods are allocated to agents in $S$, as in the core definition for $(N,G,v)$.
\begin{theorem}\label{thm:private-goods}
    For any private goods instance with linear utilities, the following are equivalent for a fractional allocation $x$.
    \begin{enumerate}
        \item $x$ is a CEEI allocation.
        \item $x$ is proportionally fair (PF).
        \item $x$ is a maximum Nash welfare (MNW) allocation.
        \item $x$ is a Lindahl allocation in the induced public goods instance.
    \end{enumerate}
    Further, all such allocations $x$ lie in the core for private goods (and hence, are Pareto optimal) and induce the same utility vector, which gives a strictly positive utility $v_i(x_i) > 0$ to every agent $i \in N$. 
\end{theorem}

\section{Allocation of Public Bads}\label{sec:public-bads}
Let us now introduce a model of public bads that mirrors the model of public goods from \Cref{sec:public-goods}. A public bads instance is given by the tuple $(N,A,c)$, where $N = [n]$ is a set of agents, $A$ is a set of $m$ alternatives, and $c = (c_1,\ldots,c_n)$ is a list of \emph{cost functions} for each agent, where $c_i : A \to \R_{\ge 0}$ for each agent $i$. For $a \in A$, write $c(a) = (c_1(a),\ldots,c_n(a))$. The cost functions induce linear costs (or disutilities) over lotteries: for $x \in \Delta(A)$, define, with slight abuse of notation, the cost to agent $i$ under $x$ as $c_i(x) \triangleq \sum_{a\in A} x(a) \cdot c_i(a)$, and the cost vector induced by $x$ as $c(x) \triangleq (c_1(x),\ldots,c_n(x))$. The (feasible) \emph{cost set} is
\[
\calC\triangleq\set{c(x) : x\in \Delta(A)} = \conv\set{c(a) : a\in A} \subseteq\R_{\ge 0}^N,
\]
where the equality is due to linearity of costs. Like for goods, $\calC$ is a non-empty compact convex polytope. We make two assumptions on the agents' cost functions. First, we assume that $0 \not\in \calC$: every alternative incurs a strictly positive cost for some agent. Secondly, we assume that for every agent $i \in N$, there exists some alternative $a \in A$ such that $c_i(a) < \max_{a' \in A}c_i(a')$: no agent is indifferent between all bads (as such an agent would have no stake in what allocation is chosen).

\subsection{Review: Allocation of Private Bads}

Our starting point is the allocation of \emph{private bads} (or \emph{chores}). The model mirrors that of private goods allocation, but we will use $B$ to denote the set of private bads (instead of $G$ for the set of private goods) and $c_{i}(b)$ to denote the cost of agent $i$ for bad $b$. A private bads instance is then given by the tuple $(N,B,c)$. We will also assume that $c_{i}(b) > 0$ for all agents $i \in N$ and bads $b \in B$; if any bad has zero cost to at least one agent, it can be allocated to any such agent and safely ignored.

\citet{BMSY17} characterize CEEI allocations in the Fisher market with private bads. To understand this, we need additional concepts.

\begin{definition}[CEEI for Private Bads]
    For private bads, we say that the pair $(x,p)$ of a fractional allocation $x$ and a price vector $p \in \R_{> 0}^B$ is a \emph{competitive equilibrium from equal incomes} (CEEI) if:
    \begin{enumerate}
    \item \emph{Minimum pain-per-buck:}
    For all $i \in N$ and $b \in B$, $x_{i}(b) > 0 \Rightarrow b \in \argmin_{b' \in B} \frac{c_{i}(b')}{p(b')}$.
    \item \emph{Budget exhaustion:} $p \cdot x_i = \nicefrac{1}{n}$ for all $i\in N$.
    \end{enumerate}
    We say that $x$ is a CEEI allocation if $(x,p)$ is a CEEI for some $p$.
\end{definition}

Hidden within their proofs, they show that all CEEI allocations of private bads satisfy a condition mirroring proportional fairness for goods. This definition applies to both private and public bads.

\begin{definition}[PF for Bads]
\label{def:pf-bads}
An allocation $x$ of bads is called \emph{proportionally fair} (PF) if we have $\frac{1}{n} \sum_{i \in N} \frac{c_i(y)}{c_i(x)} \ge 1$ for all allocations $y \in \Delta(A)$, with the convention that $\nicefrac{0}{0}=1$ and $\nicefrac{\alpha}{0} = +\infty$ for any $\alpha > 0$. Note that the direction of inequality is reversed compared to the case of goods. 
 Due to linear costs, PF is equivalent to having $\frac{1}{n}\sum_{i\in N}\frac{c_i(a)}{c_i(x)}\ge 1$ for all $a\in A$ with equality if $x(a) > 0$.\footnote{To see the equivalence, for any lottery $y$, we claim $\frac{1}{n}\sum_{i \in N}\frac{c_i(y)}{c_i(x)}
	= \sum_{a' \in \supp(y)}y(a')(\frac{1}{n}\sum_{i \in N}\frac{c_i(a')}{c_i(x)})$.
 Indeed, this is clear for all $i$ with $c_i(x) > 0$, but under our ratio convention, equality also holds when some $i$ has $c_i(x)=0$: if $c_i(y)=0$, every alternative in $\supp(y)$ has zero cost for $i$ and both sides contribute $1$, while if $c_i(y)>0$, both sides are $+\infty$. The claim and the fact that $\sum_{a' \in \supp(y)}y(a')=1$ establish sufficiency. Necessity follows because every deterministic alternative is itself a lottery.}
\end{definition}

They also note that CEEI allocations satisfy the core for private bads.
\begin{definition}[The Core for Private Bads]\label{def:core-private-bads}
An allocation $x$ of private bads is in the \emph{core} if there exist no subset of agents $S\subseteq N$ and allocation $y$ of the bads to the agents in $S$ (i.e., $y_{j}(b) = 0$ for all $j \in N\setminus S$ and $b \in B$) such that $\frac{|S|}{n} \cdot c_i(y) \le c_i(x)$ for all $i\in S$, and at least one inequality is strict.
\end{definition}

The core still implies Pareto optimality. Like in the case of goods, the existence of CEEI is known under very general conditions~\citep{SS75}, but it is no longer obtained by maximizing Nash product of costs---this would be a terrible allocation, giving bads to agents who dislike them the most---or minimizing it, which would simply give zero cost to one of the agents. Instead, \citet{BMSY17} characterize CEEI for private bads via critical points of the Nash product of costs, given by $\Prod(z) = \prod_{i \in N} z_i$, that lie on the Pareto frontier. 

\begin{definition}[Critical Points]\label{def:critical-points}
Given a smooth function $f : \calX \to \R$ and $\calY \subseteq \calX$, the critical points of $f$ in $\calY$ are given by 
\[
\Gamma(f,\calY) = \set{x \in \calY : \left(\forall y \in \calY, \nabla f(x) \cdot x \ge \nabla f(x) \cdot y \right) \lor \left(\forall y \in \calY, \nabla f(x) \cdot x \le \nabla f(x) \cdot y \right)}.
\]
Critical points include local minima, local maxima, and saddle points.
\end{definition}

\begin{definition}[Pareto Frontier]
Given a feasible cost set $\calC$, its \emph{Pareto frontier} is $\calC^{\PE} \triangleq \set{z \in \calC : \nexists z' \in \calC \text{ with } z' < z}$. These are the cost vectors induced by Pareto optimal allocations.
\end{definition}

\begin{theorem}[\citealp{BMSY17}]\label{thm:private-bads}
    For a private bads instance with linear costs, the following hold.
    \begin{enumerate}
        \item An allocation $x$ is a CEEI allocation if and only if its induced cost vector is strictly positive ($c_i(x) > 0$ for all $i \in N$), is a critical point of the Nash product $\Prod$ in $\calC$ (i.e., $c(x) \in \Gamma(\Prod,\calC)$) and lies on the Pareto frontier (i.e., $c(x) \in \calC^{\PE}$). 
        \item All CEEI allocations $x$ are proportionally fair and lie in the core (hence, they are Pareto optimal).\footnote{Technically, \citet{BMSY17} only proves that CEEI allocations are in a weak version of the core, where a group of agents can deviate if they all \emph{strictly} improve. This is due to the fact that \citet{BMSY17} studies a more general model of mixed manna (items that are goods to some agents and chores to others), where our stronger core does not always exist. For completeness, we prove that all CEEI allocations for private bads are in the core in \Cref{app:public-bads}.}
        \item There may be multiple CEEI allocations that yield different induced cost vectors.
        \item One CEEI allocation is obtained by maximizing the Nash product on the Pareto frontier, i.e., the allocations in $\argmax_{x : c(x) \in \calC^{\PE}} \prod_{i \in N} c_i(x_i)$ are CEEI. 
    \end{enumerate} 
\end{theorem}

The magical equivalence between market equilibria (CEEI for private and Lindahl for public), proportional fairness, and Nash welfare that holds for public and private goods (\Cref{thm:public-goods,thm:private-goods}) turns out to break down for private bads:
\begin{itemize}
    \item CEEI still implies proportional fairness, but the converse is no longer true, see \Cref{ex:private-bads-pf} below. Note that \citet{BMSY17} do not comment on the converse direction.
    \item CEEI is no longer obtained simply by maximizing the Nash welfare; instead, it is tied intricately to the critical points of the Nash product on the efficient frontier. Maximizing the Nash product on the efficient frontier just produces one of possibly exponentially many CEEI allocations.
    \item CEEI allocations no longer all have the same induced cost vector.
\end{itemize}

\begin{example}[PF Does Not Imply CEEI for Bads]\label{ex:private-bads-pf}
    Consider the case of two agents and a single private bad $b$, with both agents having a cost of $1$ for the bad. While the unique CEEI allocation splits the bad equally between the two agents, yielding the cost vector $(\nicefrac{1}{2},\nicefrac{1}{2})$, it is easy to check that the allocations giving the bad entirely to one of the agents, yielding cost vectors $(1,0)$ and $(0,1)$, are also proportionally fair.
    \qed
\end{example}

Despite this broken equivalence, CEEI remains an appealing rule for allocating private bads, providing fairness guarantees via proportional fairness and the core, and the characterization of \citet{BMSY17} uncovers useful structural insights about CEEI. 

\subsection{Difficulties in Extending to Public Bads}

Our goal is to establish a similar story for allocating public bads. 

\begin{quote}
    \emph{Do Lindahl equilibria yield convincingly fair allocations of public bads? Do they satisfy some version of the core and give strictly positive costs to all agents? Does the critical point characterization from private bads still hold? Does maximizing the Nash product on the Pareto frontier still yield an equilibrium?}
\end{quote}

Upon investigation, we find that allocating public bads turns out to be much trickier than the clean story so far may lead one to believe.

\subsubsection{Difficulty in Defining Lindahl Equilibrium} 

First, there is essentially no literature on defining and studying Lindahl equilibria for public bads. The first clear definition of this concept has been proposed very recently by \citet[Section 6.2]{TZ25}. They only showed that their version of Lindahl equilibrium guarantees \textit{weak} Pareto optimality (i.e., that no alternative lottery can \textit{strictly} reduce the cost to every agent simultaneously). As we show below, their version can indeed violate Pareto optimality, even for linear costs. For this reason, we term their concept \emph{weak Lindahl equilibrium}.

\begin{definition}[Weak Lindahl Equilibrium for Public Bads]
For public bads, given a lottery $x \in \Delta(A)$ and personalized prices $p = (p_i : A \to \R_{\ge 0})_{i \in N}$, we say that $(x,p)$ is a \emph{weak Lindahl equilibrium} if:
\begin{enumerate}
\item \emph{Spending constraint:} $p_i \cdot x \geq \nicefrac{1}{n}$ for all $i \in N$.
\item \emph{Cost minimization:} For every $i \in N$ and $y \in \R_{\geq 0}^A$ such that $p_i \cdot y \geq \nicefrac{1}{n}$, we have $c_i(y) \geq c_i(x)$.
\item \emph{Profit maximization:} $\sum_{i \in N} p_i(a) \geq 1$ for every $a \in A$, with equality if $x(a) > 0$.
\end{enumerate}
We say that $x$ is a weak Lindahl allocation if $(x,p)$ is a weak Lindahl equilibrium for some personalized prices $p$.
\end{definition}

\begin{example}[Weak Lindahl Equilibria Can Violate Pareto Optimality]
    Consider the instance with $2$ agents and $3$ public bads given in the table below.

\begin{table}[H]
    \centering
    \begin{tabular}{l ccc p{0.5cm} ccc p{0.5cm} ccc}
    \toprule
    & \multicolumn{3}{c}{\textbf{costs}} && \multicolumn{3}{c}{\textbf{prices}} && \multicolumn{3}{c}{\textbf{pain-per-buck}} \\
    \cmidrule(lr){2-4} \cmidrule(lr){6-8} \cmidrule(lr){10-12}
    & Bad 1 & Bad 2 & Bad 3 && Bad 1 & Bad 2 & Bad 3 && Bad 1 & Bad 2 & Bad 3 \\
    \midrule
    Agent 1 & 0 & 0 & 1 && \nicefrac{3}{4} & \nicefrac{1}{2} & \nicefrac{1}{2} && 0 & 0 & 2 \\
    Agent 2 & 3 & 4 & 4 && \nicefrac{1}{4} & \nicefrac{1}{2} & \nicefrac{1}{2} && 12 & 8 & 8 \\
    \bottomrule
    \end{tabular}
\end{table}

Clearly, the only Pareto optimal lottery is the one that places all weight on Bad 1. However, it can easily be seen that with the prices given, the lottery that places all weight on Bad 2 is a weak Lindahl allocation. This lottery will be weakly Pareto optimal.
\qed
\end{example}

Because the proportional fairness property implies Pareto optimality, this means that not all weak Lindahl equilibria are proportionally fair, so the analogy to CEEI for private bads breaks, since CEEI is always Pareto optimal. One of our contributions in \Cref{sec:bads-lindahl} is to propose a stronger definition of Lindahl equilibrium, which not only achieves PF (and hence, PO), but in fact turns out to be equivalent to PF.

\subsubsection{Difficulty in Defining the Core} 
As noted in \Cref{sec:private-goods}, the core for private goods (\Cref{def:core-private-goods}), which allows a coalition $S$ to choose an allocation $y$ of all the private goods among its members, extends naturally to the case of public goods (\Cref{def:core-public-goods}), where it can pick any lottery $y$ over the alternatives. In both cases, the outcome $y$ chosen by the coalition can be great for its members, but scaling the utilities down by $\nicefrac{|S|}{|N|}$ makes it comparable to the outcome $x$ chosen by the rule. 

The core for private bads, in contrast, \emph{forces} a coalition $S$ to pick an allocation $y$ of all the private bads among its members, which is \emph{terrible} for them. Hence, scaling their costs down by $\nicefrac{|S|}{|N|}$ makes it comparable to $x$. \emph{But how does one force a coalition to allocate ``all the bads among its members'' in the public bads model?} 

One possible formalization arises from noticing that the allocation $y$ picked by $S$ is required to impose no cost to agents in $N \setminus S$ (i.e., has zero externality). This is not always enforceable in the public bads model as there may be no such outcome. However, one can still force $S$ to bound the externality it imposes on $N \setminus S$, or rather, the factor by which the costs of its members will be scaled down can be made a function of the externality imposed. This leads us to propose the following definition of the core.

\begin{definition}[The BE Core for Public Bads]
    We say that an allocation $x \in \Delta(A)$ of public bads lies \emph{in the bounded-externality (BE) core} if there exist no $\emptyset\neq S \subseteq N$, $y \in \Delta(A)$, and $\alpha \in \R_{\ge 0}$ such that
    \begin{enumerate}
        \item $c_i(y) \le \alpha \cdot c_i(x)$ for all $i \in N\setminus S$; and \label{item:be-core-outsiders}
        \item $c_i(y) \le \frac{|N|-|N\setminus S|\cdot \alpha}{|S|} \cdot c_i(x)$ for all $i \in S$, and at least one of the inequalities is strict. \label{item:be-core-deviators}
    \end{enumerate}
\end{definition}

When $\alpha = 0$, the two conditions above reduce precisely to what the core for private bads enforces. Hence, in a private-induced public bads instance, every allocation in the BE core is also in the core for the corresponding private bads instance. That is, the BE core extends the core for private bads to public bads. We provide a formal statement and proof of this fact as \Cref{thm:be-core-private-core} in \Cref{app:bads-lindahl}.

However, the BE core may be empty. Consider a public bads instance with two agents and a single bad $b$, yielding the cost vector $(0,1)$. The outcome $x$ that selects $b$ with probability $1$ is the only feasible outcome, yet it technically fails to satisfy the BE core because $S = \set{2}$, $y = x$, and $\alpha = 0$ produce a violation of the above definition: $c_1(y) \le 0 \cdot c_1(x)$ (because $c_1(y) = 0$) and $c_2(y) < 2 \cdot c_2(x)$. In fact, this argument can be generalized to show that every outcome in the BE core must have a strictly positive cost vector. This is a desirable fairness property that CEEI satisfies for private bads (\Cref{thm:private-bads}). We discuss it in depth in \Cref{sec:bads-lindahl}.

But before that, we discuss another issue with the BE core. To understand it, let us consider the following criterion of individual fairness, applicable to both public and private bads (and to public and private goods, by using values instead of costs and flipping the direction of inequalities).  
We adopt the term ``individual fair share'' used by \citet{aziz2019fair}, but note that in the private fair division literature, an individual fairness guarantee of this sort is often referred to as ``proportionality'' which we avoid here to prevent confusion with the notion of ``proportional fairness''.

\begin{definition}[(Strong) Individual Fair Share]
    Let $c_i^{\max}$ and $c_i^{\min}$ denote the largest and smallest possible costs to agent $i$, respectively, under any outcome. We say that the allocation $x$ satisfies \emph{individual fair share} (IFS) if, for all $i \in N$,
    \[
    c_i(x) \le \frac{1}{n} \cdot c_i^{\min} + \frac{n-1}{n} \cdot c_i^{\max},
    \]
    and satisfies \emph{strong individual fair share} (sIFS) if, for all $i \in N$,
    \[
    c_i(x) \le \frac{n-1}{n} \cdot c_i^{\min} + \frac{1}{n} \cdot c_i^{\max}.
    \]
\end{definition}

For public and private goods, one can guarantee IFS via a random dictatorship, which chooses each agent's favorite outcome with probability $1/n$, but not sIFS. Strikingly, for private bads, one can guarantee sIFS via a random martyr rule, in which each agent takes on all the private bads with probability $1/n$. For public bads, sIFS is again infeasible, as the example below shows, while IFS is still guaranteed by random dictatorship. 

\begin{example}[Infeasibility of Strong Individual Fair Share for Public Bads]\label{ex:infeasible-sifs}
    \begin{wrapstuff}[r]
        \centering
        \small
        \begin{tabular}{cccc}
        \toprule
         & Bad 1 & Bad 2 & Bad 3 \\
        \midrule
        Agent 1 & $0$ & $1$ & $1$ \\
        Agent 2 & $1$ & $0$ & $1$\\
        Agent 3 & $1$ & $1$ & $0$\\
        \bottomrule
        \end{tabular}
    \end{wrapstuff}Consider the instance shown on the right
    with three agents and three bads, in which each agent has cost $0$ for a unique bad and cost $1$ for the other two bads. Every lottery assigns probability at most $\nicefrac{1}{3}$ to at least one bad, which gives the agent favoring that bad a cost of at least $\nicefrac{2}{3}$, violating sIFS.
    \qed
\end{example}

For public and private goods, the core enforces that for each agent $i \in N$:

\[v_i(x) \geq \frac{1}{n}\cdot v_i^{\max}\]

This comes from taking the deviating set $S$ to be the individual agent $\{i\}$ (for private goods, $\frac{1}{n}\cdot v_i^{\max}$ can be interpreted as $\frac{1}{n}\cdot v_i(G)$: agent $i$'s value for the entire set of items). When $v_i^{\min} = 0$ for all agents (as can be assumed in the case of private goods, where there is always at least one allocation that assigns a given agent none of the goods), this is equivalent to IFS. For public goods instances in which some agents may have $v_i^{\min} > 0$, this is a weaker guarantee than IFS. Our definition of IFS for public goods can be thought of as enforcing the above requirement after translating each agent's valuation so that its minimum is zero ($v'_i(a) = v_i(a) - v_i^{\min}$), as it is equivalent to the following:

\[v_i(x) - v_i^{\min} \geq \frac{1}{n}\cdot (v_i^{\max} - v_i^{\min})\]

Similarly, in the case of private bads, the core guarantees that for all $i \in N$, $c_i(x) \leq \frac{1}{n}\cdot c_i(B)$. Since $c_i^{\max} = c_i(B)$ and $c_i^{\min} = 0$ in the corresponding private-induced public bads instance, this exactly guarantees sIFS.

When adapting the individual fairness guarantees of the core to the public bads setting, one may think that the analogous approach would be to define a notion that also depends only on $c_i^{\max}$, such as:

\[c_i(x) \leq \frac{n-1}{n}\cdot c_i^{\max}\]

In the goods setting, the definition of IFS that also has $v_i^{\min}$ on the right-hand side makes the definition stronger. However, because of the inverted inequality in the bads case, going from the above inequality to the one enforced by IFS would actually make the requirement \emph{weaker} for instances where $c_i^{\min} > 0$. In fact, it is easy to see that the above requirement is too strong. For example, consider the instance shown in \Cref{ex:infeasible-sifs} with all the $0$ costs replaced with some $\epsilon > 0$. Clearly, in any lottery $x$ over this instance, some agent $i \in N$ must incur $c_i(x) > \frac{2}{3}$.

Thus, our definition of IFS for public bads is quite weak and is guaranteed always to exist. With this in mind, it would be natural to expect that a reasonable definition of the core for public bads should at least imply (a) IFS for all instances, and (b) the core (which implies sIFS) for private-induced instances.

In \Cref{sec:bads-lindahl}, we show that the BE core achieves (b) for a class of instances broader than private-induced instances. However, in general instances, it fails to guarantee (a). 

\begin{example}[The BE Core Can Violate Individual Fair Share]\label{ex:be-violate-ifs}
    \begin{wrapstuff}[r]
        \centering
        \small
        \begin{tabular}{ccc}
        \toprule
         & Bad 1 & Bad 2 \\
        \midrule
        Agent 1 & $1$ & $\nicefrac{1}{2}$\\
        Agent 2 & $0$ & $1$\\
        \bottomrule
        \end{tabular}
    \end{wrapstuff}Consider the instance shown on the right.
The unique BE core allocation places all weight on Bad 2. This allocation is in the BE core (by the upcoming \Cref{thm:positive-pf-be-core}, as it is PF with strictly positive costs), yet it clearly violates IFS ($c_2(x) = 1 > \nicefrac{1}{2}$).
To see uniqueness, consider any allocation $x$ placing weight $\beta \in (0,1]$ on Bad 1. The coalition $S = \{1\}$ can deviate using $y$ (all weight on Bad 1) and $\alpha = 0$. For agent 2, $c_2(y) = 0 \leq 0 \cdot c_2(x)$. For agent 1, the deviation holds because
    $ c_1(y) = 1 < 1 + \beta = 2(\tfrac{1}{2}\beta + \tfrac{1}{2}) = \frac{n}{|S|}c_1(x)$.
	\qed

\end{example}
This leads us to define an alternative definition of the core, which achieves (a), albeit it fails
(b). 

\begin{definition}[Completion Core]
    A lottery $x$ is in the \emph{completion core} if there exist no group of agents $S \subseteq N$ and lottery $y \in \Delta(A)$ such that for every agent $i \in S$, 
    \[
    \frac{|S|}{n} \cdot c_i(y) + \frac{n-|S|}{n} \cdot \max_{a \in A} c_i(a) \leq c_i(x),
    \]
    and at least one inequality is strict. 
\end{definition}

The completion core is a variant of the $\alpha$-core considered in the theory of partition function games \citep{thrall1963partitionfunction}, in which the deviating coalition $S$ assumes that the other players in $N \setminus S$ will respond in the worst possible manner for $S$. For the completion core, one can view $S$ as choosing how to allocate probability mass $\nicefrac{|S|}{|N|}$ and $N \setminus S$ as choosing how to allocate the remaining probability mass $1-\nicefrac{|S|}{|N|}$, with each $i \in S$ assuming that the coalition $N \setminus S$ will choose the outcome that is worst for agent $i$.

\begin{theorem}\label{thm:completion-core-ifs}
    Every allocation in the completion core satisfies individual fair share.
\end{theorem}
\begin{proof}
    For any allocation $x$ in the completion core, fix some agent $i$. With $S = \set{i}$, completion core guarantees that for any alternative allocation $y$:
    \[c_i(x) \leq \frac{1}{n}c_i(y) + \frac{n-1}{n}\max_{a \in A}c_i(a)\]
    Taking $y$ to be the allocation that places all the probability mass on $\argmin_{a \in A}c_i(a)$ gives the IFS guarantee for agent $i$.
\end{proof}

With some fairness definitions in hand, the natural question is: Which rules, if any, achieve these definitions? A natural place to start this search is with the rules we know work well for goods instances and private bads instances: market equilibria and proportional fairness.

\section{Lindahl Equilibria for Public Bads}\label{sec:bads-lindahl}

We begin by introducing our formalization of Lindahl equilibria for public bads.

\begin{definition}[Lindahl Equilibrium for Public Bads]\label{def:lindahl-bads}
A pair $(x,p)$ of an allocation $x \in \Delta(A)$ of public bads and personalized prices $p = (p_i : A \to \R_{\ge 0})_{i \in N}$ forms a \emph{Lindahl equilibrium} if:
\begin{enumerate}
\item \emph{Spending constraint:} $p_i \cdot x \geq \nicefrac{1}{n}$ for all $i \in N$.
\item \emph{Cost minimization:} For every $i \in N$ and $y \in \R_{\geq 0}^A$ such that $p_i \cdot y \geq \nicefrac{1}{n}$, we have $c_i(y) \geq c_i(x)$; moreover, if $c_i(y) = c_i(x) = 0$, then $\sum_{a \in A} y(a) \geq 1$.
\item \emph{Profit maximization:} $\sum_{i \in N} p_i(a) \geq 1$ for every $a \in A$, with equality if $x(a) > 0$.
\end{enumerate}
We say that $x$ is a \emph{Lindahl allocation} if $(x,p)$ is a Lindahl equilibrium for some prices $p$.
\end{definition}

The crucial change from the definition of weak Lindahl equilibrium due to \citet{TZ25} lies at the end of the cost minimization condition, where we set an additional restriction in case of $c_i(x) = 0$. This is needed to ensure Pareto optimality even when some agents receive zero cost. With this strengthened definition, we can establish a relationship between proportional fairness, market equilibrium, and the Nash welfare function, akin to that established for private bads. In particular, we will consider critical points of Nash welfare that lie in the upward closure of the cost set.

\begin{definition}[Upward Closure]
Define the \emph{upward closure} of $\calC$ as $\calCup \triangleq \set{z \in \R_{\ge 0}^N : \exists z' \in \calC \text{ with } z \geq z'}$. For $S \subseteq N$, define $\calC^S_{\ge} = \set{z_S = (z_i)_{i \in S} : (z \in \calC_{\ge}) \land (z_i = 0, \forall i \in N\setminus S)}$.
\end{definition}

We can now state our main equivalence theorem for public bads.

\begin{theorem}\label{thm:bads-pf-equiv-nonzero}
    For public bads with linear costs, the following are equivalent for any $x \in \Delta(A)$.
    \begin{enumerate}
        \item $x$ is proportionally fair (PF).
        \item $x$ is a Lindahl allocation.
        \item For $S := \set{i \in N : c_i(x) > 0}$, $c_S(x) = (c_i(x))_{i \in S}$ is a critical point in $\calC_{\ge}^S$ of the function $f_{\Nash}^S : \R_{\ge 0}^S \to \R_{\ge 0}$ given by $f_{\Nash}^S(z) = \prod_{i \in S} z_i$, i.e., $c_S(x) \in \Gamma(f_{\Nash}^S,\calC_{\ge}^S)$.
    \end{enumerate}
\end{theorem}

\noindent
The proof is split into two lemmas. The first shows the equivalence of the first two conditions.

\begin{lemma}
    A lottery $x$ is proportionally fair if and only if there exists some prices $p$ such that $(x,p)$ is a Lindahl equilibrium.
\end{lemma}
\begin{proof}
    \proofstep{PF $\Rightarrow$ Lindahl equilibrium.}
Let $x \in \Delta(A)$ be a PF lottery. Define personalized prices as follows:
    \[
        p_i(a) =
        \begin{cases}
            \frac{1}{n}\frac{c_i(a)}{c_i(x)} & \text{if } c_i(x) > 0, \\
            \frac{1}{n} & \text{if } c_i(x) = 0 \text{ and } c_i(a) = 0, \\
            1 & \text{if } c_i(x) = 0 \text{ and } c_i(a) > 0.
        \end{cases}
    \]
    We will show that $(x,p)$ is a Lindahl equilibrium.

    First, we show profit maximization. For each $a \in A$ such that $x(a) > 0$, we have $\sum_{i \in N}\frac{1}{n}\frac{c_i(a)}{c_i(x)} = 1$ from the PF conditions (\Cref{def:pf-bads}). For each $i \in N$ such that $c_i(x) > 0$, we have $p_{i}(a) = \frac{1}{n} \frac{c_i(a)}{c_i(x)}$, and for each $i \in N$ such that $c_i(x) = 0$, we must have $c_i(a) = 0$, so by our convention that $0/0 = 1$, we have $p_{i}(a) = \frac{1}{n} = \frac{1}{n}\frac{c_i(a)}{c_i(x)}$. Thus $\sum_{i \in N}p_{i}(a) = 1$ as desired.

    For each $a \in A$ such that $x(a) = 0$, for all $i \in N$ with $c_i(x) > 0$ or with $c_i(x) = 0$ and $c_i(a) = 0$ we have that $p_{i}(a) = \frac{1}{n}\frac{c_i(a)}{c_i(x)}$. If there is no agent $j$ such that $c_j(x) = 0$ and $c_j(a) > 0$, then we have that
    \[
    \sum_{i \in N}p_{i}(a) = \sum_{i \in N}\frac{1}{n}\frac{c_i(a)}{c_i(x)} \geq 1.
    \]
    If there is such an agent $j$, then we have that $p_{j}(a) = 1$, so clearly $\sum_{i \in N}p_{i}(a) \geq 1$ is true.

    For the spending constraint, for any $i \in N$ such that $c_i(x) > 0$, we have:
    \[
    p_i \cdot x = \frac{1}{n} \sum_{a \in A} \frac{c_i(a)}{c_i(x)} x(a) = \frac{1}{n} \frac{\sum_{a \in A} c_i(a) x(a)}{c_i(x)} = \frac{1}{n}.
    \]
    For any $i \in N$ with $c_i(x) = 0$, then for each $a \in A$ with $x(a) > 0$ we must have $c_i(a) = 0$, so we will have $p_{i}(a) = \frac{1}{n}$. Thus $p_i \cdot x = \sum_{a \in A}\frac{1}{n}x_a = \frac{1}{n}$.

    For cost minimization, for any $i \in N$ with $c_i(x) > 0$, the pain-per-buck ratio of $a \in A$ for $i \in N$ is
    \[
        \frac{c_i(a)}{p_{i}(a)} = n c_i(x),
    \]
    which is constant, so all lotteries that induce spending $1/n$ have the same cost, including $x$.
For any $i \in N$ such that $c_i(x) = 0$, note that since $x$ satisfies $i$'s spending constraint and has cost $0$, to establish the cost minimization condition for $i$, we need to ensure that there is no alternative lottery $y$ satisfying $i$'s spending constraint with $c_i(y) = 0$ and $\sum_{a \in A}y(a) < \sum_{a \in A}x(a) = 1$. All bads $a \in A$ with $c_i(a) = 0$ have price $p_{i}(a) = 1/n$, so clearly any alternative lottery $y$ with total price at least $\nicefrac{1}{n}$ must have size at least $1$, completing this half of the proof.

    \proofstep{Lindahl equilibrium $\Rightarrow$ PF.}
    Let $(x,p)$ be a Lindahl equilibrium. For contradiction, assume that there is an $a \in A$ such that $\frac{1}{n}\sum_{i \in N}\frac{c_i(a)}{c_i(x)} < 1$.

    First, note that for any $i \in N$, $p_i \cdot x > 1/n$ is true only if $c_i(x) = 0$. If $c_i(x) > 0$ were true, then $i$ could decrease his consumption of some $b \in A$ with $x(b) > 0$ and $c_i(b) > 0$ by some small amount without violating the spending constraint. Therefore, for all $i \in N$ with $c_i(x) > 0$, we must have $p_i \cdot x = 1/n$

    Next, note that since $x$ meets the cost minimizing condition, for any $i \in N$ with $c_i(x) > 0$, all bads $b \in A$ with $x(b) > 0$ must minimize pain-per-buck for $i$. For each agent $i \in N$, let $\text{mpb}_i$ be the minimum pain-per-buck of $i$. Then we must have that:
    \[
    c_i(a) \geq p_{i}(a)\text{mpb}_i
    \]

    Additionally, following from the fact that $p_i \cdot x = 1/n$ when $c_i(x) > 0$ and the same pain-per-buck logic, we must also have that:
    \[
    c_i(x) = \frac{1}{n}\text{mpb}_i
    \]

    Thus, for all $i \in N$ with $c_i(x) > 0$, we have:
    \[
    \frac{1}{n}\frac{c_i(a)}{c_i(x)} \geq \frac{1}{n}\frac{p_{i}(a)\text{mpb}_i}{(1/n)\text{mpb}_i} = p_{i}(a)
    \]

    For all $i \in N$ such that $c_i(x) = 0$, note that we must also have $c_i(a) = 0$. Otherwise $\frac{c_i(a)}{c_i(x)} = \infty$, which would be a contradiction to $\frac{1}{n}\sum_{i \in N}\frac{c_i(a)}{c_i(x)} < 1$.

    From this, (and recalling that we use the convention that $0/0 = 1$) we can say that for all $i \in N$ such that $c_i(x) = 0$, we must have that:
    \[
    \frac{1}{n}\frac{c_i(a)}{c_i(x)} = \frac{1}{n}\frac{0}{0} = \frac{1}{n}
    \]

    Let $S \subseteq N$ be the set of all agents $i$ such that $c_i(x) = 0$. Then, we can say the following:
    \[
    1 > \sum_{i \in N}\frac{1}{n}\frac{c_i(a)}{c_i(x)} = \sum_{i \in N \setminus S}\frac{1}{n}\frac{c_i(a)}{c_i(x)} + \sum_{i \in S}\frac{1}{n}\frac{c_i(a)}{c_i(x)} \geq \sum_{i \in N \setminus S}p_{i}(a) + \frac{|S|}{n}
    \]

    Rearranging we get:
    \[
    \sum_{i \in N \setminus S}p_{i}(a) < 1 - \frac{|S|}{n}
    \]

    If $S = \emptyset$ (all agents have positive cost for $x$), then we immediately get a contradiction to the profit maximization condition. If we have that $|S| > 0$, then in order for the profit maximization condition to be true, we must have that:
    \[
    \sum_{i \in S}p_{i}(a) > \frac{|S|}{n}
    \]

    This implies the existence of an agent $i \in S$ with $p_{i}(a) > \frac{1}{n}$. However, take the vector $y \in [0,1]^A$ where $y(a) = \frac{1}{np_{i}(a)} < 1$, and $y(b) = 0$ for every other $b \in A$. We will have that $p_i \cdot y = p_{i}(a) \frac{1}{np_{i}(a)} = \frac{1}{n}$, so $y$ satisfies the spending constraint. But this means that
    \[
    \sum_{a \in A}x(a) = 1 > \sum_{a \in A}y(a),
    \]
    contradicting the fact that $(x,p)$ meets the cost minimizing condition for $i$.
\end{proof}

The second lemma shows the equivalence between PF and suitable critical points of Nash welfare.

\begin{lemma}\label{lem:pf-characterization}
    Allocation $x$ is proportionally fair if and only if for $S := \set{i \in N : c_i(x) > 0}$, $c_S(x) = (c_i(x))_{i \in S}$ is a critical point of $f_{\Nash}^S : \R_{\ge 0}^S \to \R_{\ge 0}$ given by $f_{\Nash}^S(z) = \prod_{i \in S} z_i$ in $\calC_{\ge}^S$ (i.e., $c_S(x) \in \Gamma(f_{\Nash}^S,\calC_{\ge}^S)$).
\end{lemma}
\begin{proof}
Fix any allocation $x \in \Delta(A)$ and let $S = \set{i \in N : c_i(x) > 0}$. Define $A^S = \set{a \in A : c_i(a) = 0, \forall i \in N \setminus S}$.
The standing assumption $0\notin\calC$ ensures that $S\neq\emptyset$. Moreover, all coordinates of $c_S(x)$ are positive, so
\[
\nabla f_{\Nash}^S(c_S(x))
= f_{\Nash}^S(c_S(x))\nabla f_{\Nash}^{\log,S}(c_S(x)),
\]
where $f_{\Nash}^S(c_S(x))>0$. Thus, the two gradients differ by a positive scalar and yield the same critical-point inequalities; below we use the log-Nash gradient for convenience.

For the forward direction, assume that $x$ is PF. We show that $c_S(x) \in \Gamma(f_{\Nash}^S,\calC_{\ge}^S)$. In particular, we will show that for $w^* = \nabla f^{\log,S}_{\Nash}(c_S(x)) = \left(\frac{1}{c_i(x)}\right)_{i \in S}$, we have that $w^* \cdot z \ge w^* \cdot c_S(x) = |S|$ for all $z \in \calC^S_{\ge}$.

Due to $x$ being PF, we have that for all $y \in \Delta(A)$,
\begin{align}\label{eqn:critical-pf-mb}
    \frac{1}{n} \sum_{i \in N} \frac{c_i(y)}{c_i(x)} \ge 1.
\end{align}
Hence, the same condition holds for all $y \in \Delta(A^S)$ as well. However, $y \in \Delta(A^S)$ implies $c_i(y) = 0$ for all $i \in N \setminus S$, so $\frac{c_i(y)}{c_i(x)} = \frac{0}{0} = 1$ according to our convention in the PF definition. Plugging this into \Cref{eqn:critical-pf-mb}, we get that $\sum_{i \in S} \frac{c_i(y)}{c_i(x)} \ge |S|$. Thus, we have that $w^* \cdot z \ge |S|$ for all $z \in \calC^S = \set{z_S : z \in \calC, z_{N \setminus S} = 0}$.We need to establish the same for all $z \in \calC^S_{\ge}$. However, note that $\calC^S_{\ge} = \calC^S + \R_{\ge 0}^S = \set{z' + z'' : z' \in \calC^S, z'' \in \R_{\ge 0}^S}$. Since we have $w^* \cdot z' \ge |S|$ and $w^* \cdot z'' \ge 0$ (recall that $w^* \in \R_{> 0}^S$ and $z'' \in \R_{\ge 0}^S$), we get the desired inequality that $w^* \cdot z \ge |S|$ for all $z \in \calC^S_{\ge}$, so $c_S(x) \in \Gamma(f^S_{\Nash},\calC^S_{\ge})$.

For the reverse direction, suppose $c_S(x) \in \Gamma(f^S_{\Nash},\calC^S_{\ge})$. Note that we must have $c_S(x) \in \R_{> 0}^S$ for the gradient to be well-defined. Further, for $w^* = \nabla f^{\log,S}_{\Nash}(c_S(x)) = \left(\frac{1}{c_i(x)}\right)_{i \in S} \in \R_{> 0}^S$, we must have either
\begin{align}\label{eqn:critical-more-mb}
    w^* \cdot z \ge w^* \cdot c_S(x) = |S|, \forall z \in \calC^S_{\ge},
\end{align}
or
\begin{align}\label{eqn:critical-less-mb}
    w^* \cdot z \le w^* \cdot c_S(x) = |S|, \forall z \in \calC^S_{\ge}.
\end{align}
Noting that $c_S(x) \in \calC^S$, take $z = c_S(x) + \epsilon \in \calC^S_{\ge}$ for any $\epsilon \in \R_{> 0}^S$. Then, $w^* \cdot z = w^* \cdot c_S(x) + w^* \cdot \epsilon > w^* \cdot c_S(x) = |S|$. Hence, \Cref{eqn:critical-less-mb} cannot hold, so \Cref{eqn:critical-more-mb} must hold.

In particular, we get that for any $y \in \Delta(A^S)$, because $c_{N \setminus S}(y) = 0$, $c_S(y) \in \calC^S$. Hence, from \Cref{eqn:critical-more-mb}, we get that $\sum_{i \in S} w^*_i \cdot c_i(y) = \sum_{i \in S} \frac{c_i(y)}{c_i(x)} \ge |S|$. For each $i \in N\setminus S$, we have $\frac{c_i(y)}{c_i(x)} = \frac{0}{0} = 1$ according to our PF convention. Hence, we get
\begin{align}\label{eqn:critical-pf-2-mb}
\sum_{i \in N} \frac{c_i(y)}{c_i(x)} \ge |N|.
\end{align}
It simply remains to establish the same for all $y \in \Delta(A) \setminus \Delta(A^S)$. However, for any such $y$, there exists $i \in N\setminus S$ for which $c_i(y) > 0$, which makes $\frac{c_i(y)}{c_i(x)} = +\infty$, so \Cref{eqn:critical-pf-2-mb} holds trivially. This completes the proof.
\end{proof}

\myparagraph{Contrast between our characterization and that of \citet{BMSY17}.} Recall that \citet{BMSY17} characterize CEEI allocations of private bads as critical points of the whole Nash product $f_{\Nash} = f_{\Nash}^N$ in $\calC$ that lie on the efficient frontier $\calC^{\PE}$, i.e., $c(x) \in \calC^{\PE} \cap \Gamma(f_{\Nash},\calC)$. Our characterization differs in two crucial aspects.

First, we work with Nash product over the set of agents with positive costs. This is understandable. \citet{BMSY17} work with CEEI allocations, which yield positive costs to all agents, their condition resembles ours for the special case of $S=N$; the same holds for us in \Cref{thm:bads-pf-equiv}, where we consider allocations yielding strictly positive costs to all agents. But, more crucially, our critical points are defined in the upward closure $\calC^S_{\ge}$ (which is $\calC_{\ge}$ for $S=N$) whereas those of \citet{BMSY17} are defined in $\calC$, making our critical points more restrictive because the gradient condition needs to be enforced for a larger set of cost vectors. In the proof, we need this to deduce that the inequality in the critical point definition (\Cref{def:critical-points}) holds in the correct direction. \citet{BMSY17} get around this by instead requiring $x$ to be Pareto optimal, and stating  that $c(x) \in \calC^{\PE}$ will imply that the gradient inequality holds in the correct direction. While this works in the private bads setting, it is demonstrably false in our domain of public bads.

\begin{example}[Necessity of Critical Points in Upward Closure]\label{ex:private-bads-pf-characterization-fails}
    Consider the instance used in \Cref{ex:positive-pf-nonexistence}. We noted that the only PF solution chooses bad $b_2$ deterministically, yielding a cost vector of $(0,3)$. Let us show that the cost vector $(1,2)$ induced by bad $b_1$, which is not a PF solution and for which $S = \set{i \in N : c_i(b_1) > 0} = N$, is a critical point of $f_{\Nash}$ in $\calC$ (but not in $\calC_{\ge}$). This shows that replacing $\calC^S_{\ge}$ with $\calC^S$ and additionally demanding Pareto optimality (which $b_1$ satisfies) in \Cref{thm:bads-pf-equiv} cannot work. 

    Let $f^{\log}_{\Nash}$ be the log of $f^{\Nash}$. Because there are only two bads, $\calC$ is simply the straight line between the cost vectors of $b_1$ and $b_2$, i.e., $\calC = \set{\lambda \cdot (1,2) + (1-\lambda) \cdot (0,3) : \lambda \in [0,1]}$. At $c(b_1) = (1,2)$, we have $w^* = \nabla f^{\log}_{\Nash}(c(b_1)) = (\frac{1}{1},\frac{1}{2})$. Hence, it follows that for all $z \in \calC$, 
    \begin{align}\label{eqn:ex-critical-point-closure}
    w^* \cdot z = (1,\nicefrac{1}{2}) \cdot \lambda \cdot (1,2) + (1,\nicefrac{1}{2}) \cdot (1-\lambda) \cdot (0,3) = \lambda \cdot 2 + (1-\lambda) \cdot 1.5 \le 2,
    \end{align}
    which means $c(b_1) \in \Gamma(f_{\Nash},\calC)$, but $b_1$ is not a PF solution. The problem arises because the inequality in \Cref{eqn:ex-critical-point-closure} goes in the opposite direction than what we need to establish PF.  
    \qed
\end{example}

\begin{proposition}\label{prop:deterministic-pf}
    There always exists a deterministic, proportionally fair allocation. Specifically, an allocation that deterministically selects an alternative which (a) maximizes the number of agents with zero cost and (b), subject to that, minimizes the product of costs of the agents with positive costs is guaranteed to be proportionally fair.
\end{proposition}
\begin{proof}
    Let $a$ be such a deterministically selected alternative. We need to show that it satisfies PF, which means that $\frac{1}{n}\sum_{i \in N}\frac{c_i(a')}{c_i(a)} \geq 1$ for all $a' \in A$ (see \Cref{def:pf-bads}).

    For contradiction, assume that there exists some $a' \in A$ such that $\sum_{i \in N}\frac{c_i(a')}{c_i(a)} < n$.
    Let $S = \set{i \in N : c_i(a) > 0}$. The standing assumption $0\notin\calC$ ensures that $S\neq\emptyset$. We must have $c_i(a') = 0$ for all $i \in N \setminus S$; otherwise we would have $\sum_{i \in N}\frac{c_i(a')}{c_i(a)} = +\infty$. It must also be the case that $c_i(a') > 0$ for all $i \in S$; otherwise, the fact that $a$ maximizes the number of agents with zero cost would be contradicted.
    Thus, by our convention that $\nicefrac{0}{0} = 1$, we must have that
    \begin{equation}
    	\label{eq:deterministic-pf-contradictor}
    	\sum_{i \in S}\frac{c_i(a')}{c_i(a)} < |S|.
    \end{equation}
    But by our definition of $a$, we must have $0 < \prod_{i \in S}c_i(a) \leq \prod_{i \in S}c_i(a')$. Dividing by $\prod_{i \in S}c_i(a)$ gives:
    \[1 \leq \prod_{i \in S}\frac{c_i(a')}{c_i(a)},\]
    and then using the AM--GM inequality we can derive that
    \[1 = 1^{1/|S|} \leq \left(\prod_{i \in S}\frac{c_i(a')}{c_i(a)}\right)^{1/|S|} \leq \frac{1}{|S|}\sum_{i \in S}\frac{c_i(a')}{c_i(a)},\]
    which contradicts \eqref{eq:deterministic-pf-contradictor}.
\end{proof}

The equivalence between PF and Lindahl allocations, the close relationship to critical points of the Nash product in \Cref{thm:bads-pf-equiv}, and the guaranteed existence of a deterministic solution from \Cref{prop:deterministic-pf} are particularly appealing. However, previous examples have already established that an allocation simply being PF is not an indication of its fairness. Recall from \Cref{ex:private-bads-pf} that in private-induced instances, all allocations corresponding to CEEI give all agents strictly positive costs, while there may be PF allocations that give some agents $0$ cost at the expense of other agents.

\subsection{Strictly Positive Allocations}\label{sec:strictly-positive}
The above discussion leads us to a refinement of the previous theorem based on ensuring that each agent receives strictly positive costs --- we term such allocations \emph{strictly positive allocations}.

\begin{definition}[Zero-Respecting Lindahl Equilibrium]
    A Lindahl equilibrium $(x,p)$ is \emph{zero-respecting} if for all $i \in N$ and $a \in A$, $c_i(a) = 0$ implies that $p_i(a) = 0$.
\end{definition}

\begin{theorem}\label{thm:bads-pf-equiv}
    For any $x \in \Delta(A)$, the following are equivalent.
    \begin{enumerate}
        \item $c(x) \gg 0$ and $x$ is proportionally fair (PF). \label{item:bpe-pf}
        \item There exist personalized prices $p \in \R_{\ge 0}^{N \times A}$ for which $(x,p)$ is a zero-respecting Lindahl equilibrium. \label{item:bpe-lindahl}
        \item $c(x) \gg 0$ and $c(x)$ is a critical point of the Nash welfare $f_{\operatorname{Nash}} : \R_{\ge 0}^N \to \R_{\ge 0}$ given by $f(y) = \prod_{i \in N} y_i$ in $\calC_{\geq}$. \label{item:bpe-critical}
    \end{enumerate}
\end{theorem}
\begin{proof}
\proofstep{(\ref{item:bpe-pf}) $\Leftrightarrow$ (\ref{item:bpe-critical}).}
    This equivalence clearly follows from \Cref{thm:bads-pf-equiv-nonzero}, as we are simply imposing the restriction of $c(x) \gg 0$ for both conditions.

    \proofstep{(\ref{item:bpe-pf}) $\Leftrightarrow$ (\ref{item:bpe-lindahl}).}
    To complete the proof, we will show that a Lindahl equilibrium $(x,p)$ will be zero-respecting if and only if $c_i(x) > 0$ for all $i \in N$.

    To see this, note that for any Lindahl equilibrium $(x,p)$, and for any agent $i \in N$, if $c_i(x) > 0$, then it cannot be the case that there exists some $a \in A$ with $c_i(a) = 0$ and $p_{i}(a) > 0$. If this were the case, then $i$ would have a pain-per-buck ratio of $0$ for $a$, and thus, any lottery whose support does not exclusively include bads that have $c_i(a) = 0$ would violate the cost-minimizing condition of the Lindahl equilibrium.

    Conversely, for any zero-respecting Lindahl equilibrium $(x,p)$, we must have that $c_i(x) > 0$ for all $i \in N$. If this were not true, and $c_i(x) = 0$ for some $i \in N$, then all the bads $a$ in the support of $x$ would have to have $c_i(a) = 0$. By the zero-respecting property, we have $p_{i}(a) = 0$ for all such goods, and thus $(x,p)$ would violate the spending constraint of the Lindahl equilibrium with respect to $i$.
\end{proof}

Restricting our attention to the set of strictly positive PF allocations, desirable fairness properties begin to emerge. First, we show that strictly positive PF allocations lie in the BE core, something that certainly was not true for all PF allocations.

\begin{theorem}\label{thm:positive-pf-be-core}
    Every strictly positive PF allocation is in the bounded-externality core.
\end{theorem}
\begin{proof}
    For contradiction, assume this is false. There exists an allocation $x \in \Delta(A)$ such that $c_i(x) > 0$ for all $i \in N$ and $x$ is PF, but not in the bounded-externality core.

    Thus, there exists some $S \subseteq N, y \in \Delta(A), \alpha \in \R_{\geq 0}$ such that:
    \begin{enumerate}
        \item $c_i(y) \le \alpha \cdot c_i(x)$ for all $i \in N\setminus S$; and \label{item:be-core-violation-outsiders}
        \item $c_i(y) \le \frac{|N|-|N\setminus S|\cdot \alpha}{|S|} \cdot c_i(x)$ for all $i \in S$, and at least one of the inequalities is strict. \label{item:be-core-violation-deviators}
    \end{enumerate}
    These inequalities straightforwardly imply the following:
    \[\sum\nolimits_{i \in N}\frac{c_i(y)}{c_i(x)} = \sum\nolimits_{i \in N \setminus S}\frac{c_i(y)}{c_i(x)} + \sum\nolimits_{i \in S}\frac{c_i(y)}{c_i(x)} < 
    |N \setminus S| \cdot \alpha + (n - |N \setminus S| \cdot \alpha) = n \]
    The strictness of the inequality comes from the strictness of one of the inequalities from condition~(\ref{item:be-core-violation-deviators}). This implies that $x$ is not a PF allocation, a contradiction, which completes the proof.
\end{proof}

As mentioned earlier, under private-induced public bads instances, BE core implies the core for private bads. This means that in all such instances, if $y$ is a strictly positive PF allocation, and $x$ is a private bads allocation that induced $y$, then $x$ will be in the core.
\citet{KP25} asked whether there are connections between Lindahl equilibrium and CEEI.
We can show that enforcing strict positivity is enough to establish a perfect correspondence between PF allocations of public bads and CEEI allocations of private bads in private-induced instances. This gives a connection between Lindahl equilibria and CEEI that is quite different from the duality connection recently established by \citet{TZ25}.

\begin{theorem}\label{thm:lindahl-fisher-equiv}
    Let $x$ be any allocation in a private bads instance $(N,B,c)$ and $y$ be its induced allocation in the induced public bads instance $\calR(N,B,c)$. Then, $x$ is a CEEI allocation for $(N,B,c)$ if and only if $y$ is a strictly positive proportionally fair allocation for $\calR(N,B,c)$.
\end{theorem}
\begin{proof}
    We will start with the forward direction, which follows easily from past results. If $(x,p)$ is a CEEI for the private bads instance $(N,B,c)$, and $y$ is an allocation over $\calR(N,B,c)$ that is induced by $x$, then there exists personalized prices $q$ such that $(y,q)$ is a zero-respecting Lindahl equilibrium for $\calR(N,B,c)$.
By construction, the private bads instance $(N,B,c)$ and its corresponding public bads instance $\calR(N,B,c)$ will have an identical cost space $\calC$. We know from \citet{BMSY17} that all CEEIs will be strictly positive PF solutions with respect to the cost space of the private bads instance. Therefore, since $x$ and $y$ have the same cost vector, $y$ must also be strictly positive PF with respect to the cost space $\calC$.
From \Cref{thm:bads-pf-equiv}, we know that this implies $y$ is a zero-respecting Lindahl allocation.

    Now we prove the other direction. If $(y,q)$ is a zero-respecting Lindahl equilibrium in $\calR(N,B,c)$, and $x$ is a private allocation over $(N,B,c)$ that induces $y$, then there exists prices $p$ such that $(x,p)$ is a CEEI.
Let $\calA$ be the set of all $n^m$ possible private integral allocations $z$ over $(N,B,c)$.
For each $i \in N$, let $\text{mpb}_i = \min_{z \in \calA : q_i(z) > 0}\frac{c_i(z)}{q_i(z)}$ be the minimum pain-per-buck value of agent $i$ under the equilibrium $(y,q)$.
    Thus, for all $z \in \calA$ with $q_i(z) > 0$, we have that $\frac{c_i(z)}{q_i(z)} \geq \text{mpb}_i$, and if $y(z) > 0$, then either $\frac{c_i(z)}{q_i(z)} = \text{mpb}_i$,  or $c_i(z) = 0$ and $q_i(z) = 0$.
For each $b \in B$, define the price $p(b) = \min_{i \in N}\frac{c_i(b)}{\text{mpb}_i}$. We will show that $(x,p)$ is a CEEI.

    First note the following fact: for each $z \in \mathcal{A}$ such that $y(z) > 0$, $z_i(b) = 1$ for any $i \in N$ implies that $i \in \argmin_{i \in N}\frac{c_i(b)}{\text{mpb}_i}$. For contradiction, assume this fact is not true, so there is some $z \in \mathcal{A}$ with $y(z) > 0$, and some $b^* \in B$ and $i^*, i^{**} \in N$ such that $z_{i^*}(b^*) = 1$ but $\frac{c_{i^*}(b^*)}{\text{mpb}_{i^{*}}} > \frac{c_{i^{**}}(b^*)}{\text{mpb}_{i^{**}}}$. 
    
    Consider the allocation $z'$, where all goods except for $b^*$ are allocated the same as in $z$, but $b^*$ is allocated to $i^{**}$ rather than $i^*$.
    Since $y(z) > 0$, $z$ must be profit maximizing among all private integral allocations in $\mathcal{A}$. Specifically, we must have $1 = \sum_{i \in N}q_{i}(z) \leq \sum_{i \in N}q_{i}(z')$. From $y(z) > 0$, it also follows that for all $i \in N$ such that $q_{i}(z) > 0$, we have
    \[\frac{c_{i}(z)}{q_{i}(z)} = \frac{\sum_{b : z_i(b) = 1}c_{i}(b)}{q_{i}(z)} = \text{mpb}_i.\]
    Thus we must have:
    \[\sum_{i \in N}\sum_{b: z_i(b) = 1}\frac{c_{i}(b)}{\text{mpb}_i} \leq \sum_{i \in N}\sum_{b \in z'_i}\frac{c_{i}(b)}{\text{mpb}_i}.\]
    However, because $z'$ is just $z$ with $b^*$ reallocated, we can say that
    \[\sum_{i \in N}\sum_{b : z'_i(b) = 1}\frac{c_{i}(b)}{\text{mpb}_i} = \sum_{i \in N}\sum_{b : z_i(b) = 1}\frac{c_{i}(b)}{\text{mpb}_i} - \frac{c_{i^*}(b^*)}{\text{mpb}_{i^*}} + \frac{c_{i^{**}}(b^*)}{\text{mpb}_{i^{**}}}.\]
    Plugging this into the profit maximization inequality gives $\frac{c_{i^*}(b^*)}{\text{mpb}_{i^*}} \leq \frac{c_{i^{**}}(b^*)}{\text{mpb}_{i^{**}}}$, a contradiction, which proves the fact.

    Now we will show that $(x,p)$ is a CEEI. The first step is to show that each agent is spending their entire budget, i.e., for each $i \in N, \sum_{b \in B}x_i(b)p(b) = \nicefrac{1}{n}$. From the fact that $y$ is induced by $x$, we have that $x_{i}(b) = \sum_{z \in \mathcal{A}}y(z)z_{i}(b)$ for all $i \in N, b \in B$. This gives us that for all $i \in N$,
\[ \sum_{b \in B}x_{i}(b)p(b) = \sum_{b \in B}p(b)\sum_{z \in \mathcal{A}}y(z)z_{i}(b) = \sum_{z \in \mathcal{A}}y(z)\sum_{b \in B}z_{i}(b)p(b).\]
    We can show that for each $i \in N$ and $z \in \mathcal{A}$ with $y(z) > 0$, we have $q_{i}(z) = \sum_{b \in B}z_{i}(b)p(b) = \sum_{b : z_i(b) = 1}p(b)$. This is due to the fact that for each $b$ with $z_i(b) = 1$, we have $i \in \argmin_{i \in N}\frac{c_{i}(b)}{\text{mpb}_i}$. Using this together with the definition of the private goods price vector $p$, for each $i \in N$, we have
\[\sum_{b : z_i(b) = 1}p(b) = \sum_{b : z_i(b) = 1}\frac{c_{i}(b)}{\text{mpb}_i} = \frac{c_{i}(z)}{\text{mpb}_i} = q_{i}(z).\]
    Plugging this into the previously derived equation gets us:
\[\sum_{b \in B}x_{i}(b)p(b) = \sum_{z \in \mathcal{A}}y(z)\sum_{b \in B}z_{i}(b)p(b) = \sum_{z \in \mathcal{A}}y(z)q_{i}(z) = \nicefrac{1}{n},\]
    where the last equality follows from the fact that $(y,q)$ is a zero-respecting Lindahl equilibrium, and thus each agent pays exactly $\nicefrac{1}{n}$ for $y$.

    As a final step, we must prove that for each $i \in N$ and $b \in B$ with $x_{i}(b) > 0$, $b$ minimizes $i$'s pain-per-buck in the private bads instance.
    For contradiction, assume this is false, so there is some $i \in N$ and $b, b' \in B$ such that $x_{i}(b) > 0$ and $\frac{c_{i}(b)}{p(b)} > \frac{c_{i}(b')}{p(b')}$. Because $x_{i}(b) > 0$, there must be some $z \in \mathcal{A}$ such that $y(z) > 0$ and $z_i(b) = 1$. This implies that $i \in \argmin_{i \in N}\frac{c_{i}(b)}{\text{mpb}_i}$, thus $p(b) = \frac{c_{i}(b)}{\text{mpb}_i}$ and $\text{mpb}_i = \frac{c_{i}(b)}{p(b)}$.
    Thus we must have $\text{mpb}_i > \frac{c_{i}(b')}{p(b')}$. But this would imply that $p(b') > \frac{c_{i}(b')}{\text{mpb}_i}$, a contradiction to how we define $p$. This completes the proof.
\end{proof}

These two theorems initially seem very promising. Is ``strictly positive PF'' the ultimate fair solution concept for allocation of public bads? For private-induced instances, it coincides with CEEI, which is arguably the ultimate solution. However, there are very simple non-private-induced instances, in which strictly positive PF exhibits rather undesirable behaviors. 

\myparagraph{Problem of non-existence.} To begin with, strictly positive PF allocations may not always exist.

\begin{example}[Strictly Positive PF Allocations May Not Exist, \citealp{mariotti2005nash}]\label{ex:positive-pf-nonexistence}
    \begin{wrapstuff}[r]
        \centering
        \small
        \begin{tabular}{ccc}
        \toprule
         & Bad $b_1$ & Bad $b_2$ \\
        \midrule
        Agent 1 & $1$ & $0$ \\
        Agent 2 & $2$ & $3$ \\
        \bottomrule
        \end{tabular}
    \end{wrapstuff}Consider the instance shown on the right. Note that there exists a Pareto optimal solution (namely, picking $b_1$ deterministically), which induces a strictly positive cost vector. Despite this, we can show that the unique PF lottery $x^*$ puts probability $1$ on $b_2$ and yields zero cost to agent $1$:
It is easy to check that $x^*$ is PF because $\frac{1}{2} \cdot (\frac{c_1(b_1)}{c_1(b_2)} + \frac{c_2(b_1)}{c_2(b_2)}) = \frac{1}{2} \cdot (+\infty + \frac{2}{3}) > 1$. On the other hand, for any other lottery $x$ that puts probability $\alpha \in (0,1]$ on $b_1$ and $1 - \alpha$ on $b_2$, we have
$\frac{1}{2} \big(\frac{c_1(x^*)}{c_1(x)} + \frac{c_2(x^*)}{c_2(x)}\big) = \frac{1}{2} \cdot \frac{3}{3-\alpha} < 1.
    $
\qed
\end{example}

\myparagraph{Problem of completion core failure.} Even when strictly positive PF allocations exist, they do not lie in the completion core, which can lead to very counterintuitive behavior on extremely simple instances.

\begin{example}
	\begin{wrapstuff}[r]
		\centering
		\small
		\begin{tabular}{ccc}
			\toprule
			& Bad 1 & Bad 2 \\
			\midrule
			Agent 1 & $1$ & $0$ \\
			Agent 2 & $1$ & $0$ \\
			Agent 3 & $0$ & $1$ \\
			\bottomrule
		\end{tabular}
	\end{wrapstuff}On the instance shown on the right, it can be checked that the unique strictly positive PF allocation places probability $\nicefrac{2}{3}$ on Bad 1 and $\nicefrac{1}{3}$ on Bad 2. 
    This is in contrast to the unique allocation in the completion core, which places probability $\nicefrac{1}{3}$ on Bad 1 and $\nicefrac{2}{3}$ on Bad 2. The latter allocation makes much more sense for this extremely simple binary instance, as it is what would be achieved if we gave each agent $\nicefrac{1}{3}$ probability mass and allowed them to place it on the bad they did not disapprove of.
    
    In fact, if one extrapolates this example to $k$ agents with cost $1$ for bad $1$ and $0$ for bad $2$, and still a single agent with cost $0$ for bad $1$ and $1$ for bad $2$, the unique strictly positive PF allocation will place probability $\nicefrac{k}{k+1}$ on bad $1$ and $\nicefrac{1}{k+1}$ on bad $2$. Thus, the larger the cohesive group of agents preferring bad $2$ becomes, the \emph{weaker} it becomes, in contrast to the traditional proportional representation principle, which demands that cohesive groups get stronger as they get larger. In \Cref{app:impossibilities}, we provide a discussion of a fairness notion called \emph{participation} that captures this specific undesirable behavior.
    \qed
\end{example}

\myparagraph{Problem of IFS failure.} In \Cref{ex:be-violate-ifs}, we demonstrated that allocations in the BE core can fail to achieve even the weakest form of IFS, which can be guaranteed for public bads. Since all strictly positive PF allocations lie in the BE core, they fail to achieve this simple fairness guarantee as well. 

Given these undesirable behaviors of strictly positive PF allocations, why are they so compelling in private-induced instances? What structural property of such instances makes them magically achieve the fairness properties of CEEI allocations, such as even sIFS? Can we isolate this structural property and extend the positive results to a broader class of public bads instances? The answer is yes!

\subsection{Axes-Cutting Instances}\label{sec:axes-cutting}

We now introduce a special class of instances, generalizing private-induced instances, in which strictly positive PF allocations turn out to be well-behaved.

\begin{definition}[Axes-Cutting Instance, \citealp{mariotti2005nash}]
    A public bads instance $(N,A,c)$ is an \emph{axes-cutting} instance if, for every agent $i \in N$, there exists an alternative $a_i \in A$ such that $c_{i}(a_i) > 0$ and $c_{i'}(a_i) = 0$ for all $i' \in N \setminus \set{i}$.
\end{definition}

In words, axes-cutting requires that for each agent, there exists an option where they become a martyr, bearing all the cost. The name comes from the fact that the feasible cost set $\calC$ of such an instance cuts (or rather touches) the $n$ axes at $c(a_1),\ldots,c(a_n)$. 

This restriction was originally proposed by \citet{mariotti2005nash}, who studied a related version of public selection where a cost vector must be selected from an arbitrary convex cost space $\calC$. They studied a solution concept called the ``Nash rationing solution'', which can be shown to be equivalent to strictly positive PF solutions in our setting. Facing a similar problem of non-existence in general instances, they proposed the axes-cutting restriction as a means to guarantee existence. Specifically, they show that maximizing the Nash product of costs subject to Pareto optimality, a solution concept that is guaranteed to be CEEI on private bads instances, turns out to be strictly positive PF on all axes-cutting instances. For completeness, we provide a direct proof of the following result without establishing equivalence with their definition. 

\begin{theorem}[\citealp{mariotti2005nash}]\label{thm:axes-cut-existence}
    For any axes-cutting public bads instance, a strictly positive PF allocation always exists. Specifically, there exists a strictly positive PF allocation whose cost vector maximizes the Nash product of costs over the Pareto frontier.
\end{theorem}

The following lemmas establish the proof.

\begin{lemma}[Convexity of Nash Superlevel Sets]
For $\lambda\ge 0$, define the Nash superlevel set
\begin{equation}
D_\lambda \triangleq \{z\in \R_{\ge 0}^N : f_{\Nash}(z)\ge \lambda\}.
\label{eq:def-Dlambda}
\end{equation} 
For every $\lambda>0$, the set $D_\lambda$ is closed and convex, and $D_\lambda\subseteq \R_{>0}^N$.
\end{lemma}
\begin{proof}
For $\lambda>0$, we have
\[
D_\lambda = \{z\in \R_{>0}^N : \sum_{i \in N} \log z_i \ge \log \lambda\}.
\]
Because the sum of logarithms is concave on $\R_{>0}^N$, its superlevel sets are convex. Closedness follows from the continuity of $f_{\Nash}$ on $\R_{\ge 0}^N$, and $D_\lambda\subseteq \R_{>0}^N$ because $f_{\Nash}(z)=0$ whenever any $z_i = 0$.
\end{proof}

\begin{lemma}[Large Nash Contours Lie in Upward-Closure]
\label{lem:Dlambda-in-Cge}
Let $\alpha_i = c_i(a_i)$, where $a_i$ is the axes-cutting alternative corresponding to agent $i$ that yields zero cost to all other agents. Then, under the axes-cutting condition, we have $D_\lambda \subseteq \calC_{\ge}$ for every $\lambda \ge (\sum_{i=1}^n \alpha_i/n)^n$.
\end{lemma}
\begin{proof}
The axes-cutting condition implies that $\alpha_i e_i=c(a^{(i)})\in \calC\subseteq \calC_{\ge}$.
Since $\calC_{\ge}$ is upward closed, if $z\in \R_{\ge 0}^N$ satisfies $z_i\ge \alpha_i$ for any $i \in N$, then $z\in\calC_{\ge}$.
Now, fix $z\in D_\lambda$. By the AM--GM inequality,
\[
\frac{\sum_{i \in N} z_i}{n}\ \ge\ \Bigl(\prod_{i \in N} z_i\Bigr)^{1/n}\ \ge\ \lambda^{1/n}.
\]
When $\lambda \ge (\sum_{i=1}^n \alpha_i/n)^n$, this implies $\sum_{i \in N} z_i \ge \sum_{i \in N} \alpha_i$. Hence, there must exist some $i \in N$ with $z_i\ge \alpha_i$, which implies $z\in \calC_{\ge}$.
\end{proof}

The next lemma provides a self-contained proof of a supporting vectors argument used by \citet{mariotti2005nash} in their original proof of this theorem.

\begin{lemma}[Supporting Vectors for Upward-Closed Convex Sets are Nonnegative]
\label{lem:nonneg-normal}
Let $K\subseteq \R_{\ge 0}^N$ be closed, convex, and upward closed (i.e., $z\in K$ and $z'\in \R_{\ge 0}^N$ implies $z+z'\in K$).
If $z^\star\in \partial K$, then there exists a nonzero $w\in \R_{\ge 0}^N$ such that
\begin{equation}
w \cdot z \ \ge\ w\cdot z^\star, \forall z\in K.
\label{eq:support-upward}
\end{equation}
\end{lemma}
\begin{proof}
Because $z^\star\in\partial K$ and $K$ is closed, there exists a sequence $(z^t)_{t\ge 1}$ with $z^t\notin K$ and $z^t\to z^\star$. By the separating hyperplane theorem applied to the disjoint closed convex sets $K$ and $\{z^t\}$, for each $t$ there exists a nonzero vector $w^t$ and a scalar $\beta^t$ such that
\[
w^t\cdot z \ \ge\ \beta^t \ >\ w^t\cdot z^t, \forall z\in K.
\]
Normalize $w^t$ so that $\|w^t\|_2=1$. By compactness of the unit sphere, pass to a convergent subsequence (still indexed by $t$) with $w^t \to w$. Taking limits along this subsequence and using $z^t\to z^\star$ yields \Cref{eq:support-upward} for some nonzero $w$.

It remains to show $w\in \R_{\ge 0}^N$. Suppose for contradiction that $w_k<0$ for some coordinate $k$.
Fix any $z\in K$. By upward closure, $z+t e_k\in K$ for all $t\ge 0$. Then \Cref{eq:support-upward} gives
\[
w\cdot (z+t e_k)\ \ge\ w\cdot z^\star \quad \text{for all } t\ge 0,
\]
i.e., $(w\cdot z) + t w_k \ge w\cdot z^\star$ for all $t\ge 0$, which is impossible when $w_k<0$ and $t\to\infty$. Hence, $w\geq 0$.
\end{proof}

\begin{lemma}[Existence of PF with Strictly Positive Costs; Nash Maximization on $\calC^{\PE}$]
\label{thm:existence-and-maxNash}
In any axes-cutting public bads instance, there exists an allocation $x^\star\in\Delta(A)$ such that:
\begin{enumerate}
\item $c(x^\star)\in \R_{>0}^N$ and $x^\star$ is PF.
\item The vector $z^\star\triangleq c(x^\star)$ maximizes the Nash product of costs over the Pareto frontier $\calC^{\PE}$:
\begin{equation}
z^\star \in \argmax_{z\in \calC^{\PE}} \ \prod_{i \in N} z_i,
\qquad\text{and}\qquad
\prod_{i \in N} z^*_i > 0.
\label{eq:maxNash-on-PE}
\end{equation}
\end{enumerate}
\end{lemma}
\begin{proof}
\proofstep{Step 1: Define the critical Nash level $\lambda^\star>0$.}
By Lemma~\ref{lem:Dlambda-in-Cge}, there exists $\bar\lambda>0$ such that $D_{\bar\lambda}\subseteq \calC_{\ge}$.
Define
\begin{equation}
\lambda^\star \triangleq \inf\{\lambda\ge 0 : D_\lambda \subseteq \calC_{\ge}\}.
\label{eq:def-lambda-star}
\end{equation}
Then $\lambda^\star \le \bar\lambda$.
Because $D_\lambda$ is decreasing in $\lambda$ and $\calC_{\ge}$ is closed, we have $D_{\lambda^\star}\subseteq \calC_{\ge}$.

We claim $\lambda^\star>0$. By our assumption that $0\notin \calC$, since $\calC$ is compact, there exists $\varepsilon>0$ such that
\begin{equation}
[0,\varepsilon]^n \cap \calC=\emptyset.
\label{eq:box-away}
\end{equation}
By upward closure, this implies $[0,\varepsilon]^n \cap \calC_{\ge}=\emptyset$ as well.
Let $\lambda_0\triangleq \varepsilon^n>0$. Then the vector $\varepsilon\mathbf{1}\in D_{\lambda_0}$ but $\varepsilon\mathbf{1}\notin \calC_{\ge}$, so $D_{\lambda_0}\nsubseteq \calC_{\ge}$.
Hence $\lambda^\star\ge \lambda_0>0$.

\proofstep{Step 2: There exists $z^\star \in D_{\lambda^\star}\cap \partial\calC_{\ge}$ and $z^\star \gg 0$.}
If $D_{\lambda^\star}\subseteq \mathrm{int}(\calC_{\ge})$, then the compact set $[0,\textstyle\sum_i\alpha_i]^N\setminus \mathrm{int}(\calC_{\ge})$ is disjoint from $D_{\lambda^\star}$. Hence, by continuity of $f_{\Nash}$, there exists $\eta \in (0, \lambda^\star)$ such that
\[
D_{\lambda^\star-\eta}\cap [0,\sum\nolimits_{\,i} \alpha_i]^N \ \subseteq\ \calC_{\ge}.
\]
Moreover, any $z\in D_{\lambda^\star-\eta}$ lying outside of $[0,\sum_i \alpha_i]^N$ satisfies $z_i > \sum_{j}\alpha_j \geq \alpha_i$ for some $i \in N$, and therefore belongs to $\calC_{\ge}$ by the same argument as \Cref{lem:Dlambda-in-Cge}. Thus $D_{\lambda^\star-\eta}\subseteq \calC_{\ge}$, contradicting the definition of $\lambda^\star$ as an infimum. Therefore $D_{\lambda^\star}\cap \partial \calC_{\ge}\neq\emptyset$; fix any
\begin{equation}
z^\star\in D_{\lambda^\star}\cap \partial \calC_{\ge}.
\label{eq:choose-xstar}
\end{equation}
Since $\lambda^\star>0$, we have $z^\star\in \R_{>0}^N$.

\proofstep{Step 3: A supporting vector $w^\star$ exists and is proportional to $(1/z^\star)$.}
Apply \Cref{lem:nonneg-normal} to $K=\calC_{\ge}$ at $z^\star$ to obtain a nonzero $w^\star\geq 0$ such that
\begin{equation}
w^\star\cdot z \ \ge\ w^\star\cdot z^\star, \forall z\in \calC_{\ge}.
\label{eq:wstar-support-Cge}
\end{equation}
Since $D_{\lambda^\star}\subseteq \calC_{\ge}$, \Cref{eq:wstar-support-Cge} holds in particular for all $z\in D_{\lambda^\star}$, so $z^\star$ minimizes $w^\star\cdot z$ over $D_{\lambda^\star}$.

We claim $z^\star\in \partial D_{\lambda^\star}$, i.e., $f_{\Nash}(z^\star)=\lambda^\star$. Indeed, if $f_{\Nash}(z^\star)>\lambda^\star$, then $z^\star$ lies in the interior of $D_{\lambda^\star}$ (because $f_{\Nash}$ is continuous), so there exists $\delta>0$ and a vector $u\in\R^N$ with $w^\star\cdot u<0$ such that $z^\star+\tau u\in D_{\lambda^\star}$ for all $\tau\in[0,\delta]$, contradicting minimality of $z^\star$ for the linear objective $w^\star\cdot z$ over $D_{\lambda^\star}$. Thus,
\begin{equation}
f_{\Nash}(z^\star)=\lambda^\star.
\label{eq:Fxstar}
\end{equation}

Because $z^\star\in \R_{>0}^N$ and $D_{\lambda^\star}=\{z\in \R_{>0}^N:f_{\Nash}(z)\ge \lambda^\star\}$, the KKT conditions for minimizing a linear function over the convex set $D_{\lambda^\star}$ imply that there exists $\mu>0$ such that
\begin{equation}
w^\star = \mu \nabla f^{\log}_{\Nash}(z^\star)=\mu\Bigl(\frac{1}{z^\star_1},\ldots,\frac{1}{z^\star_n}\Bigr).
\label{eq:wstar-grad}
\end{equation}
In particular, $w^\star\in \R_{>0}^N$.

\proofstep{Step 4: $z^\star\in \calC^{\PE}$ and it yields a PF allocation.}
First, we show $z^\star\in \calC$. By definition of $\calC_{\ge}$, there exists $z\in\calC$ with $z\leq z^\star$.
If $z\neq z^\star$, then $z < z^\star$, and since $w^\star\in\R_{>0}^N$ we have $w^\star\cdot z < w^\star\cdot z^\star$,
contradicting \Cref{eq:wstar-support-Cge} because $z\in\calC\subseteq\calC_{\ge}$.
Hence $z=z^\star\in\calC$.

Next, $z^\star$ is Pareto efficient in $\calC$: if there were $z\in\calC$ with $z < z^\star$, then $w^\star\cdot z < w^\star\cdot z^\star$ (since $w^\star\gg 0$), contradicting \Cref{eq:wstar-support-Cge}.
Thus, $z^\star\in\calC^{\PE}$.

Finally, since $z^\star\in \calC$, there exists $x^\star\in\Delta(A)$ with $c(x^\star)=z^\star$.
Using \Cref{eq:wstar-grad,eq:wstar-support-Cge} gives
\[
\nabla f^{\log}_{\Nash}(z^\star)\cdot z \ \ge\ \nabla f^{\log}_{\Nash}(z^\star)\cdot z^\star, \forall z\in \calC_{\ge}.
\]
This is exactly the necessary critical point condition with $z^\star=c(x^\star)$. Hence, by \Cref{thm:bads-pf-equiv-nonzero}, $x^\star$ is PF and has strictly positive costs.

\proofstep{Step 5: $z^\star$ maximizes the Nash product of costs on $\calC^{PE}$.}
We already have $f_{\Nash}(z^\star)=\lambda^\star>0$ from \Cref{eq:Fxstar}. Let $z\in \calC^{PE}$.
If $f_{\Nash}(z)>\lambda^\star$, choose $\lambda$ with $\lambda^\star<\lambda<f_{\Nash}(z)$.
Then, $z\in D_\lambda$, and moreover $z\in \mathrm{int}(D_\lambda)$ because $f^{\log}_{\Nash}(z)>\log\lambda$.
Hence, there exists $z'\in D_\lambda$ with $z' < z$ (e.g., $z'=z-\epsilon\mathbf{1}$ for a sufficiently small $\epsilon>0$). Since $D_\lambda\subseteq \calC_{\ge}$ (by definition of $\lambda^\star$), we have $z'\in\calC_{\ge}$, so there exists $z''\in\calC$ with $z''\leq z' < z$.
Thus $z'' < z$, contradicting Pareto efficiency of $z$. Therefore $f_{\Nash}(z)\le \lambda^\star=f_{\Nash}(z^\star)$ for all $z\in\calC^{\PE}$, proving \Cref{eq:maxNash-on-PE}.
\end{proof}

\Cref{thm:axes-cut-existence} already highlights the power of axes-cutting instances, as they allow us to get around the non-existence of strictly positive PF in general instances, shown by examples such as \Cref{ex:positive-pf-nonexistence}. Combining this with \Cref{thm:positive-pf-be-core}, this also shows the existence of outcomes in the bounded-externality core for such instances. However, their utility extends even further. We can also show that the structure of axes-cutting instances allows us to circumvent other counterintuitive behavior exhibited by strictly positive PF and the BE core, leading to very compelling definitions of fairness in this setting and rules for achieving them. We can now show that under axes-cutting instances, the BE core implies sIFS and the completion core.

\begin{theorem}\label{thm:PF-sifs-axes-cut}
    In all axes-cutting instances, any allocation that is in the bounded-externality core satisfies strong individual fair share.
\end{theorem}
\begin{proof}
    If $n=1$, every allocation satisfies sIFS because $c_1(x)\le c_1^{\max}$, so the result is immediate. So suppose that $n\ge 2$. For every agent $j$, an axes-cutting alternative for any other agent has zero cost to $j$, and therefore $c_j^{\min}=0$.
For the sake of contradiction, assume that the result is false and there is an allocation $x$ in the bounded-externality core that does not satisfy sIFS. Thus, there exists some $j \in N$ such that $c_j(x) > \frac{1}{n}\max_{a \in A}c_j(a)$. Let $a_j \in A$ be the axes-cutting bad for agent $j$ (i.e., the bad such that $c_{i}(a_j) = 0$ for all $i \neq j$). Then, clearly, we must have $c_j(x) > \frac{1}{n}c_j(a_j)$ as well.

    Now, consider the potential BE core deviating group with $S = \set{j}$, $\alpha = 0$, and $y$ being the lottery that places all its weight on bad $a_j$.
For each $i \in N \setminus S$, we have that $c_i(y) = 0 \leq 0 \cdot c_i(x)$, so the inequalities in condition~(\ref{item:be-core-outsiders}) of a BE core deviating group are satisfied.
For agent $j$, we have that
    \[\frac{n - |N \setminus S|\alpha}{|S|}c_j(x) = nc_j(x) > c_j(a_j) = c_j(y),\]
    where the second transition follows from the sIFS violation. Therefore, condition~(\ref{item:be-core-deviators}) for a BE core deviation holds, and thus $x$ cannot be in the BE core, a contradiction.
\end{proof}

\begin{theorem}\label{thm:sifs-completion-core}
    Any allocation that is Pareto optimal and satisfies strong individual fair share is in the completion core.
\end{theorem}
\begin{proof}
    If a lottery $x$ satisfies sIFS, then for each $i$ we have $c_i(x) \leq \frac{1}{n} \max_{a \in A}c_i(a) + \frac{n - 1}{n}\min_{a \in A}c_i(a)$. Now, fix some strict subset of agents $S \subsetneq N$ and some lottery $y$. For all $i \in S$,
    \[c_i(x) \leq \frac{1}{n}\max_{a \in A}c_i(a) + \frac{n-1}{n}\min_{a \in A}c_i(a) \leq \frac{n-|S|}{n}\max_{a \in A}c_i(a) + \frac{|S|}{n}\min_{a \in A}c_i(a) \leq \frac{n - |S|}{n}\max_{a \in A}c_i(a) + \frac{|S|}{n}c_i(y).\]
    Hence, when $|S| < n$, no group $S, y$ can serve as a completion core deviating group.
    Finally, if $|S| = n$, the deviating condition requires $c_i(y)\le c_i(x)$ for every $i\in N$, with at least one strict inequality, which is impossible as $x$ is Pareto optimal.
\end{proof}

Putting the above results together, along with the fact that all strictly positive PF solutions will be Pareto optimal, we get the following:
\begin{theorem}
    In axes-cutting instances, all strictly positive PF solutions will satisfy strong individual fair share, and will be in the completion core.
\end{theorem}

However, this analysis of axes-cutting instances still leaves open a very important question: How do we handle instances that are not axes-cutting? We have already established that strictly positive PF solutions are not the solution in these cases. This is what we will address in the next section.

\section{Fairness Beyond Axes-Cutting Instances}\label{sec:impossibilities}

Beyond axes-cutting instances, we know that fairness properties such as sIFS may cease to exist, and that even when strictly positive PF solutions do exist, they may fail to achieve basic fairness notions. To begin our search for fairness in general instances, we highlight exactly which properties of strictly positive PF solutions are causing the unfairness by introducing an axiomatic characterization of strictly positive PF solutions, first explored by \citet{mariotti2005nash} in the axes-cutting setting.

In previous sections, we introduced the notation of $\calC$ as the cost set of a public bads problem, and $\calC_{\geq}$ as the upward-closure of that cost set. When speaking about the cost sets of a specific instance $I$, we will refine this notation as $\calC(I)$ and $\calC_{\geq}(I)$, respectively. Additionally, a public bads allocation rule $\phi$ maps each instance $I = (N,A,c)$ to a possibly empty set $\phi(I) \subseteq \Delta(A)$ of lotteries over the alternative set $A$.

Finally, recall that for a fixed lottery $x$, $c(x) \in \mathbb{R}^N$ denotes the cost vector of that lottery. For any rule $\phi$ and any instance $I$, let $c(\phi(I)) = \{c(x) : x \in \phi(I)\}$.

\begin{definition}[Cost consistency]
    A rule $\phi$ is \emph{cost consistent} if for any instance $I$ and any pair of lotteries $x,x'$ over $I$ with the property that $c_i(x) = c_i(x')$ for all $i \in N$, $x \in \phi(I)$ if and only if $x' \in \phi(I)$.
\end{definition}

\begin{definition}[Weak symmetry]
    Define an \emph{identity} instance to be an instance with $n$ agents and $n$ bads, where for each agent $i \in N$, there exists a unique bad $a_i \in A$ such that $c_i(a_i) = 1$ and $c_{i'}(a_i) = 0$ for all $i' \neq i$.

    A rule $\phi$ is \emph{weakly symmetric} if for any identity instance, $\phi$ returns a single allocation, the allocation which places weight $\nicefrac{1}{n}$ on each bad.
\end{definition}

\begin{definition}[Scale-freeness]
    A rule $\phi$ is \emph{scale-free} if for any two instances $I = (N,A,c)$ and $I' = (N,A,c')$, where for some agent $i \in N$, there exists a constant $k>0$ such that $c'_i(a) = kc_i(a)$ for all $a \in A$, and $c'_{i'}(a) = c_{i'}(a)$  for all $i' \neq i$, then $\phi(I) = \phi(I')$.
\end{definition}

\begin{definition}[Lower contraction consistency]
    A rule $\phi$ satisfies \emph{lower contraction consistency} if, for any pair of instances $I = (N,A,c)$ and $I' = (N,A',c')$ such that $\calC_{\geq}(I) \subseteq \calC_{\geq}(I')$, we have

    \[c'(\phi(I')) \cap \calC(I) \subseteq c(\phi(I)).\]

    In words, if some lottery $x \in \phi(I')$ induces a cost vector $c'(x)$ that is feasible in $I$ (i.e., $c'(x) \in \calC(I)$), then $\phi(I)$ must select some lottery that induces the same cost vector.
\end{definition}

\citet{mariotti2005nash} take similar formulations of the axioms above (defined slightly differently for their general setting that is only concerned with cost vectors rather than lotteries over sets of alternatives), and show that for all axes-cutting instances, the rule which returns all strictly positive PF allocations satisfies the axioms of cost consistency, weak symmetry, scale-freeness, and lower contraction consistency. Furthermore, they show that the cost vectors returned by \emph{any} rule that satisfies these axioms must be a superset of the strictly positive PF cost vectors. Their proof is based on the original proof of \citet{Nash50b}, showing that MNW is characterized by a similar set of axioms in the positive-utility version of this problem. It is not difficult to see that both directions of their proof easily generalize to non-axes-cutting instances as well. For the sake of completeness, we include full versions of these proofs in \Cref{app:impossibilities}.

With these results in hand, we can next leverage impossibility results to identify which aspects of strictly positive PF cause it to be unfair on non-axes-cutting instances.

\begin{theorem}
    No rule satisfying weak symmetry and lower contraction consistency returns only allocations that satisfy individual fair share.
\end{theorem}
\begin{proof}
    Consider some rule $\phi$ that satisfies weak symmetry and lower contraction consistency, and fix some $n \geq 2$. By weak symmetry, $\phi$ will choose the uniform lottery over all $n$ bads when run on the identity instance $I$ for $n$ agents.

    \begin{wrapstuff}[r]
    \scalebox{0.88}{\begin{tabular}{lcc}
    \toprule
              & Bad $1$ & Bad $2$ \\
    \midrule
    Agent $1$ & $1$     & $\nicefrac{1}{n}$ \\
    Agent $2$ & $0$     & $\nicefrac{1}{n}$  \\
    $\vdots$  & $\vdots$ & $\vdots$ \\
    Agent $n-1$ & $0$     & $\nicefrac{1}{n}$ \\
    Agent $n$ & $0$     & $\nicefrac{1}{n}$   \\
    \bottomrule
    \end{tabular}}
    \end{wrapstuff}
    Next, consider the instance $I'$ shown on the right with $n$ agents and $2$ bads.
    It is easy to see that $\calC_{\geq}(I') \subseteq \calC_{\geq}(I)$, as the cost vector of Bad $1$ is achievable in both $I$ and $I'$, and Bad $2$ from $I'$ has a cost vector $(\nicefrac{1}{n}, \nicefrac{1}{n}, \dots, \nicefrac{1}{n})$, which can be achieved by the uniform lottery over all $n$ options in $I$. Therefore, by lower contraction consistency, $\phi(I')$ must select some lottery $y$ which has the cost vector $(\nicefrac{1}{n}, \nicefrac{1}{n}, \dots, \nicefrac{1}{n})$. The only such lottery places all weight on Bad $2$.

    However, for a lottery $y$ to satisfy IFS, for any agent $i \in N \setminus \set{1}$ we must have $c_i(y) \leq \frac{n-1}{n}\frac{1}{n} < \frac{1}{n}$, which is not true for $y$.
\end{proof}

Looking at this impossibility, if one regards IFS as a necessary fairness requirement for a rule, either weak symmetry or lower contraction consistency must be dropped. It is reasonable to assume that weak symmetry is the more desirable property. Any rule that does not satisfy it would return very counterintuitive results on identity instances. Luckily, there is a very straightforward rule that drops lower contraction consistency while achieving the other three axioms and, on all instances (axes-cutting or otherwise), always finds an allocation in the completion core, which implies IFS.

\begin{definition}[Flipped-MNW Solution]
    For a given public bads instance with agents $N$ and public bads $A$, construct a dual public goods instance over the same agents and bads, where for each $i \in N$, $i$'s valuation function over $A$ is given by $v_i(a) = \max_{a' \in A}c_i(a') - c_i(a)$ for all $a \in A$. The Flipped-MNW solution returns the set of all lotteries that maximize Nash welfare over the agents in the dual public goods instance.
\end{definition}

Note that a similar rule, called the ``Dual Nash Bargaining Solution'', was also briefly considered by \citet{mariotti2005nash} because it satisfies all axioms other than lower contraction consistency (though again, their rule was formulated specifically for axes-cutting instances, and for a model where only cost vectors were considered).

\begin{theorem}\label{thm:flipped-MNW-completion-core}
    All lotteries returned by Flipped-MNW are in the completion core.
\end{theorem}
\begin{proof}
    For a contradiction, assume that $x$ is a Flipped-MNW lottery for some public bads division instance, and it is not in the completion core, meaning that there exists some nonempty $S \subseteq N$ and an alternative lottery $y$ such that for every $i \in S$,
    \[\frac{|S|}{n}c_i(y) + \frac{n-|S|}{n}\max_{a \in A}c_i(a) \leq c_i(x)\]
    with the inequality being strict for some $i' \in S$.

    First, note that for every agent $i \in N$, there exists an alternative $a' \in A$ such that $c_i(a') < \max_{a \in A}c_i(a)$. This implies that $\max_{a \in A}c_i(a) - c_i(a') > 0$. Thus, in the corresponding dual goods instance that Flipped-MNW constructs, every agent $i \in N$ must have some good $a'$ such that $v_i(a') > 0$. This meets the non-degeneracy condition we established for public goods problems, and thus, by \Cref{thm:public-goods}, the lottery $x$ lies in the core for this dual public goods instance.
    This means that, for $S$ and $y$ in particular, it cannot be the case that for each $i \in S$, $\frac{|S|}{n}v_i(y) \geq v_i(x)$, with the inequality being strict for $i'$.
    
    However, we can rearrange the completion core violation inequality to derive a contradiction. Using the definition of the dual valuation function, we get
    \[\frac{|S|}{n}\max_{a \in A}c_i(a) - \frac{|S|}{n}v_i(y) + \frac{n-|S|}{n}\max_{a \in A}c_i(a) \leq \max_{a \in A}c_i(a) - v_i(x).\]
    The $\max$ terms on the two sides cancel, yielding, for all $i \in S$,
    \[\frac{|S|}{n}v_i(y) \geq v_i(x),\]
    with the inequality being strict for $i' \in S$, a contradiction.
\end{proof}

While Flipped-MNW has advantages over strictly positive PF in general instances, it is worth noting that on private-induced public bads instances, Flipped-MNW can give unintuitive results, for example, by failing to find an allocation that satisfies sIFS even though one is always possible on such instances. An example of this is included in \Cref{app:impossibilities}.

With this, we were able to establish two rules for the division of public bads. Strictly positive PF achieves fairness guarantees under axes-cutting and private-induced instances, but beyond those settings, may fail to satisfy the completion core or IFS, or even fail to exist at all. Flipped-MNW is not nearly as strong on axes-cutting instances or private instances, but it maintains the fairness guarantee of the completion core for all instances.

\section{Computation}\label{sec:computation}

How can we compute PF allocations in the public bads setting?
It is easy to find \emph{some} PF allocation because, as we have seen in \Cref{prop:deterministic-pf}, there always exists a deterministic PF allocation (i.e., one with weight 1 placed on a single bad). Thus, we can iterate through all $m$ deterministic allocations and check which are PF. Checking whether a given allocation satisfies PF is easy because, due to linear costs, this reduces to checking that $\frac{1}{n}\sum_{i\in N}\frac{c_i(a)}{c_i(x)}\ge 1$ for all $a\in A$.

However, we may also be interested in finding non-deterministic PF allocations, for example to compute strictly positive allocations (which correspond to zero-respecting Lindahl equilibria, \Cref{sec:strictly-positive}). For this task, the following characterization is helpful.

\begin{proposition}\label{prop:pf-supports}
	For any $x \in \Delta(A)$, write $S = \supp(x)$ and write $N(S) = \{ i \in N : c_i(a) > 0 \text{ for some } a \in S \}$. Then the following are equivalent.
	\begin{enumerate}
		\item $x$ is proportionally fair (PF), \label{item:pfsupp-pf}
		\item we have $x_S \in \argmax_{y \in \Delta(S)} \prod_{i \in N(S)} c_i(y)$ and for all $a \in A \setminus S$,  $\frac{1}{n} \sum_{i \in N} \frac{c_i(a)}{c_i(x)} \ge 1$. \label{item:pfsupp-nash}
	\end{enumerate}
\end{proposition}
\begin{proof}
	\proofstep{(\ref{item:pfsupp-pf}) $\Rightarrow$ (\ref{item:pfsupp-nash}).}
	Suppose $x$ is PF. \Cref{def:pf-bads} immediately gives us the second part of (\ref{item:pfsupp-nash}), i.e., that for all $a \in A \setminus S$,  $\frac{1}{n} \sum_{i \in N} \frac{c_i(a)}{c_i(x)} \ge 1$. It also tell us that for every $a \in S$, we must have $\frac{1}{n} \sum_{i \in N} \frac{c_i(a)}{c_i(x)} = 1$. Therefore, for every allocation $y$ with $\supp(y) \subseteq S$ we have $\frac{1}{n} \sum_{i \in N} \frac{c_i(y)}{c_i(x)} = 1$ and thus in particular $\frac{1}{n} \sum_{i \in N} \frac{c_i(y)}{c_i(x)} \le 1$ and also $\frac{1}{n} \sum_{i \in N(S)} \frac{c_i(y)}{c_i(x)} \le 1$ (because for every $i \in N \setminus N(S)$ we have $c_i(y) = 0$). Therefore, with the alternative set restricted to $S$ and the agent set restricted to $N(S)$, the allocation $x_S$ satisfies the public-goods version of PF (\Cref{def:pf-public-goods}) with costs interpreted as valuations. Hence, by \Cref{thm:public-goods}, $x_S$ maximizes Nash welfare among agents $N(S)$ over $\Delta(S)$ with costs interpreted as valuations, which gives us the first part of (\ref{item:pfsupp-nash}).
	
	\proofstep{(\ref{item:pfsupp-nash}) $\Rightarrow$ (\ref{item:pfsupp-pf}).}
	Suppose that $x_S \in \argmax_{y \in \Delta(S)} \prod_{i \in N(S)} c_i(y)$. Then by \Cref{thm:public-goods}, $x_S$ satisfies the public-goods version of PF and so $\sum_{i \in N(S)} \frac{c_i(a)}{c_i(x)} \le |N(S)|$ for all $a \in S$. In fact, this holds with equality, since if it was $< |N(S)|$ for some $a \in S$ (for which by definition of $S$ we have $x(a) > 0$), then 
	\[
	\textstyle
	|N(S)| = \sum_{i \in N(S)} \frac{c_i(x)}{c_i(x)} = \sum_{a \in S} x(a)  \sum_{i \in N(S)} \frac{c_i(a)}{c_i(x)} < |N(S)|,
	\] 
	a contradiction. Thus, for all $a \in S$, we have $\sum_{i \in N(S)} \frac{c_i(a)}{c_i(x)} = |N(S)|$. In addition, for each $i \in N \setminus N(S)$, we have $c_i(a) = c_i(x) = 0$ and thus $\frac{c_i(a)}{c_i(x)} = 1$ by our ratio convention. Hence $\frac{1}{n} \sum_{i \in N} \frac{c_i(a)}{c_i(x)} = 1$ for all $a \in S$. Together with the assumption of (\ref{item:pfsupp-nash}) about $a \in A \setminus S$, we in fact have $\frac{1}{n} \sum_{i \in N} \frac{c_i(a)}{c_i(x)} \ge 1$ for all $a \in A$, establishing that $x$ satisfies PF.
\end{proof}

The characterization implies that for each support, if there exists a PF allocation with that support, it can be obtained by maximizing the Nash product of costs.
Since there are only $2^m - 1$ supports and since by strict concavity the maximizer of a Nash product is unique in costs, we can deduce the following.

\begin{corollary}
	Let $C_{\text{PF}} = \{ c(x) : x \in \Delta(A) \text{ s.t. $x$ is PF} \}$ be the set of cost vectors of all PF allocations in a given instance. Then $|C_{\text{PF}}| \le 2^m - 1$.
\end{corollary}

\Cref{prop:pf-supports} suggests a simple exponential-time algorithm for enumerating all essentially different PF allocations, at least approximately: iterate through all supports $S$, approximately compute the Nash optimizer $y$ on $S$ using convex programming, and check whether $\supp(y) = S$ (i.e., it places strictly positive weight on each member of $S$) and whether $y$ satisfies the PF condition with respect to $A \setminus S$ (as in \Cref{prop:pf-supports}(2)).

While this algorithm is exhaustive and practical for analyzing small instances, it does not scale to cases with many alternatives. To efficiently compute an approximate strictly positive PF allocation, we follow a technique developed by \citet{chaudhury2024competitive} for computing CEEI in private bads instances. Their approach is based on deriving an optimization program with a similar structure to that of the dual of the Nash product maximization program and then arguing that its KKT points are equivalent to CEEI allocations. A similar dual program can be written down for the case of public bads, by analogy to the public goods dual program derived by \citet[Theorem 5]{KP25}.

\begin{equation}
	\label{eq:dual}
	\begin{aligned} 
		\underset{{\beta \ge 0 } }{\text{maximize}}\ & -\sum_{i\in N} \tfrac{1}{n} \log\beta_i \\ 
		\text{subject to} \ & \sum_{i\in N} c_i(a) \beta_i \ge 1 \quad \text{for all $a \in A$}
	\end{aligned}
\end{equation}
The standing assumption $0\notin\calC$ ensures that this program is feasible. Indeed, it implies $\sum_{i\in N}c_i(a)>0$ for every $a\in A$, so setting all $\beta_i$ to any sufficiently large common value satisfies every constraint.

We show below that the KKT points of this program correspond exactly to PF allocations $x$ with $c(x) \gg 0$, or equivalently to zero-respecting Lindahl equilibrium allocations (see \Cref{thm:bads-pf-equiv}).
\begin{proposition}\label{prop:kkt-points}
	The KKT points $(\beta, x)$ of the optimization program \eqref{eq:dual} are in one-to-one correspondence with allocations $x \in \Delta(A)$ that are PF and satisfy $c(x) \gg 0$.
\end{proposition}
\begin{proof}
	Write $x = (x_a)_{a \in A}$ for the dual variables corresponding to the constraints of program \eqref{eq:dual}. Suppose $(\beta, x)$ is a KKT point of the program. That means that $\beta$ is feasible and the following conditions hold.
	\begin{enumerate}
		\item \emph{Feasibility and domain constraint:} $\beta$ is a feasible solution to program \eqref{eq:dual} and $\beta_i > 0$ for all $i$ by the implicit domain constraint of the log objective. \label{item:kkt-domain-constraint}
		\item \emph{Dual feasibility:} $x_a \ge 0$ for all $a \in A$. \label{item:kkt-dual-feasibility}
		\item \emph{Stationarity:} $\frac{1}{n\beta_i} = \sum_{a \in A} c_{i}(a) x_a$ for each $i \in N$. \label{item:kkt-stationarity}
		\item \emph{Complementary slackness:} $x_a( \sum_{i\in N} c_i(a) \beta_i - 1) = 0$. \label{item:kkt-compl-slackness}
	\end{enumerate}
	Define personalized prices $p_i(a) = c_i(a) \beta_i$. We claim that $(x, p)$ is a zero-respecting Lindahl equilibrium (see \Cref{def:lindahl-bads}) and hence that $x$ is a strictly positive PF allocation (\Cref{thm:bads-pf-equiv}).
	
	\emph{Zero-respecting:} Follows immediately from the definition of $p_i(a)$.
	
	\emph{Profit maximization:} From feasibility, $\sum_{i \in N} p_i(a) \ge 1$, and by complementary slackness (\ref{item:kkt-compl-slackness}), if $x_a > 0$ then $\sum_{i \in N} p_i(a) = 1$.
	
	\emph{Spending constraint:} We have
	\[
	\sum_{a \in A} p_i(a) x_a = \beta_i \sum_{a \in A} c_i(a) x_a \overset{(\ref{item:kkt-stationarity})}{=} \frac1n.
	\]
	
	\emph{Sum of $x$:} We have $x \in \Delta(A)$ because, using profit maximization and summing the spending-constraint equality above,
	\[
	\sum_{a \in A} x_a = \sum_{a \in A} \left( \sum_{i \in N} p_i(a) \right) x_a = \sum_{i \in N} \sum_{a \in A} p_i(a) x_a = n \cdot \frac1n = 1.
	\]
	
	\emph{Cost minimization:} Let $y \in \mathbb{R}^A_{\ge 0}$ such that $p_i \cdot y \ge 1/n$. Then
	\[
	c_i(y) = \sum_{a \in A} c_i(a) y_a = \frac{1}{\beta_i} \sum_{a \in A} p_i(a) y_a \ge \frac{1}{n\beta_i} \overset{(\ref{item:kkt-stationarity})}{=} \sum_{a \in A} c_{i}(a) x_a = c_i(x).
	\]
	The clause ``$c_i(y) = c_i(x) = 0$'' of \Cref{def:lindahl-bads} never occurs in a zero-respecting Lindahl equilibrium (because $c_i(y) = 0$ implies that for every $a \in A$ with $y_a > 0$ we have $p_i(a) = 0$, so $p_i \cdot y = 0 \not\ge \frac1n$). \\
	
	Conversely, suppose that $x$ is a strictly positive PF allocation. Define $\beta_i = 1/(n c_i(x))$. We claim that $(\beta, x)$ is a KKT point of the program. Indeed, for (\ref{item:kkt-domain-constraint}) feasibility follows from PF since $\sum_{i\in N} c_i(a) \beta_i = \frac{1}{n} \sum_{i \in N} \frac{c_i(a)}{c_i(x)} \ge 1$, and the implicit domain constraint holds by strict positivity; (\ref{item:kkt-dual-feasibility}) holds because $x \in \Delta(A)$; (\ref{item:kkt-stationarity}) holds by definition of $\beta$ and of $c_i(x)$; and complementary slackness (\ref{item:kkt-compl-slackness}) holds by the PF condition which holds with equality on its support.
\end{proof}
As we discussed in \Cref{sec:bads-lindahl}, zero-respecting Lindahl allocations do not always exist, and thus the program may not have KKT points. Therefore, we again consider the axes-cutting assumption under which existence of such allocations is guaranteed (\Cref{thm:axes-cut-existence}). Under this assumption, one can check that the objective value of \eqref{eq:dual} is bounded, and thus the program has no ``poles''. This means that we can apply the Greedy Frank-Wolfe algorithm proposed by \citet{chaudhury2024competitive} on the program \eqref{eq:dual} to compute an approximate zero-respecting Lindahl equilibrium. That algorithm is an iterative first-order method that in each update step only needs to solve a linear program; thus can be run in polynomial time. We give the details below.

\begin{theorem}\label{thm:frank-wolfe}
	For an axes-cutting instance, for each $i \in N$ let $\alpha_i = c_i(a_i)$, where $a_i \in A$ is the alternative that only imposes costs on $i$, and let $\delta = \min \{ c_i(a) : i \in N, a \in A, c_i(a) > 0 \}$ be the smallest positive cost in the instance. Then, for every $0 < \epsilon \le 1$, the Greedy Frank-Wolfe algorithm finds an $\epsilon$-approximate zero-respecting Lindahl equilibrium in $O(\frac{\sum_i\log(\alpha_i/\delta)}{\epsilon^2})$ iterations.
\end{theorem}
\begin{algorithm}[t]
	\caption{Greedy Frank--Wolfe algorithm for solving program \eqref{eq:log-type-convex-max-prob}}
	\label{alg:GFW}
	\DontPrintSemicolon
	\KwIn{An initial feasible point $x^0 \in \mathcal{X}$; stopping tolerance $\eta > 0$.}
	\For{$t = 1,2,\ldots$}{
		$x^{t} \gets \arg\max_{x\in \mathcal X}\ \smash{-\sum_{i \in [n]} a_i \frac{x_i}{x^{t-1}_i}} $\;
		\If{$\sum_i \big( \frac{x^{t}_i}{x^{t-1}_i} - 1 \big)^2 \leq \eta^2$}{
			\Return{$x^t$}\;
		}
	}
\end{algorithm}

We prove this result using the methods of \citet{chaudhury2024competitive} who study the following optimization problem, of which program \eqref{eq:dual} is a special case:
\begin{equation}
	\begin{aligned}
		\underset{x \in \R^n_{>0}}{\text{maximize}}\ & g(x) := -\sum_{i \in [n]} a_i \log{x_i} \\ 
		\text{subject to} \ & x \in \mathcal{X},
	\end{aligned}
	\label{eq:log-type-convex-max-prob}
\end{equation}
where $a_i > 0$ for each $i \in [n]$ and $\mathcal{X} \subseteq \mathbb R^n_{>0}$ is a nonempty closed convex set. 

\citet{chaudhury2024competitive} prove the following convergence bound, generalizing results of \citet{journee2010generalized}.

\begin{theorem}[\citealp{chaudhury2024competitive}, Theorem 6]\label{thm:gfw-convergence-rate}
	Suppose we solve program~\eqref{eq:log-type-convex-max-prob} using the Greedy Frank-Wolfe algorithm (\Cref{alg:GFW}) starting from a feasible point $x^0$. Assume we are given a finite upper bound $g^* < +\infty$ on the value of the objective function of \eqref{eq:log-type-convex-max-prob}. Then, for any $\eta > 0$, we can find a pair of iterates $(x^t, x^{t + 1})$ such that $\sum_i \big( \frac{x^{t+1}_i}{x^t_i} - 1 \big)^2 \leq \eta^2$ after at most 
	\[
		t \le \left\lceil\frac{3(g^* - g(x^0))}{\min_i a_i} \frac{1}{\eta^2} + \frac{g^* - g(x^0)}{c \min_i a_i}\right\rceil
	\]
	iterations, where $c = \sqrt{3} - 1 - \frac{1}{2}\log{3} > 0$. 
	\label{crl:stepsize-convergence-for-log-type-convex-max}
\end{theorem}

Thus, to obtain our result, we need to show that close consecutive iterates encode approximate Lindahl equilibria. We first define the formal sense in which we are approximating the concept of Lindahl equilibrium.

\begin{definition}[Approximate Zero-Respecting Lindahl Equilibrium]
A pair $(x,p)$ of an allocation $x \in \Delta(A)$ of public bads and personalized prices $p = (p_i : A \to \R_{\ge 0})_{i \in N}$ forms an $\epsilon$-approximate zero-respecting \emph{Lindahl equilibrium} if:
\begin{enumerate}
\item \emph{Zero-respecting:} For every $i \in N$ and $a \in A$, $p_i(a) = 0$ whenever $c_i(a) = 0$.
\item \emph{Approximate spending constraint:} $\nicefrac{1}{n} - \epsilon \leq p_i \cdot x \leq \nicefrac{1}{n} + \epsilon$ for all $i \in N$.
\item \emph{Cost minimization:} For every $i \in N$ and $y \in \R_{\geq 0}^A$ such that $p_i \cdot y \geq p_i \cdot x$, we have $c_i(y) \geq c_i(x)$.
\item \emph{Profit maximization:} $\sum_{i \in N}p_i(a) \geq 1$ for every $a \in A$, with equality if $x(a) > 0$.
\end{enumerate}
\end{definition}
Note that only the spending constraint is approximated. The cost-minimization condition is correspondingly adapted to compare $x$ with alternative allocations $y$ whose spending is at least $p_i \cdot x$.

\begin{lemma}\label{lem:convergence-gives-approx-lindahl}
	Suppose the instance satisfies axes-cut. Consider consecutive iterates $\beta^{t-1}$ and $\beta^t$ computed by the Greedy Frank-Wolfe algorithm applied to \eqref{eq:dual}. If, for some $0 < \eta \le 1/2$, they satisfy $|\beta_i^t/\beta_i^{t-1} - 1| \le \eta$ for every $i \in N$, then the dual variables of the LP for computing $\beta^t$ encode a $4\eta$-approximate zero-respecting Lindahl equilibrium.
\end{lemma}
\begin{proof}
	The Greedy Frank-Wolfe algorithm solves the following LP in each step:
	\[
	\underset{\beta \ge 0}{\text{min}} \:
	\sum_{i \in N} \frac1n \frac{\beta_i}{\beta^{t-1}_i} 
	\quad \text{s.t.} \quad
	\sum_{i\in N} c_i(a) \beta_i \ge 1 \quad \text{for all $a \in A$}.
	\]
	Its LP dual is
	\[
	\underset{x \ge 0}{\text{max}} \:
	\sum_{a \in A} x_a
	\quad \text{s.t.} \quad
	\sum_{a\in A} c_i(a) x_a \le \frac{1}{n \beta_i^{t-1}} \quad \text{for all $i \in N$}.
	\]
	Let $x^t$ denote the optimum dual variables corresponding to the iterate $\beta^t$. Set $S = \frac{1}{n}\sum_i \frac{\beta_i^t}{\beta_i^{t-1}}$ to be the objective value of the primal (and by strong duality also of the dual). Define personalized prices $p_i^t(a) = c_i(a)\beta_i^t$ for all $i \in N$ and $a \in A$ and define the allocation $\bar x^t = x^t / S$. Note that $\bar x^t \in \Delta(A)$ because $S = \sum_{a\in A} x_a$ by strong duality.
	
	We claim that $(\bar x^t, p^t)$ forms an approximate zero-respecting Lindahl equilibrium.
	
	\emph{Zero-respecting:} Immediate from the definition of the personalized prices.
	
	\emph{Profit maximization:} Note that $\sum_{i \in N} p_i^t(a) \ge 1$ follows from feasibility in the primal, and from complementary slackness we get that if $x^t_a > 0$ then $\sum_{i \in N} p_i^t(a) = 1$.
	
	\emph{Approximate spending constraint:} Let $i \in N$. Let $a_i \in A$ be an alternative with $c_i(a_i) > 0$ and $c_j(a_i) = 0$ for all $j \neq i$, whose existence is guaranteed by the axes-cut condition. Note that by primal feasibility $c_i(a_i) \beta^t_i \ge 1$ and hence $\beta^t_i > 0$.
	Thus, by complementary slackness, the $i$th constraint of the dual holds with equality, so we have $c_i(x^t) =  \frac{1}{n\beta_i^{t-1}}$. Thus
	\[
	\sum_{a \in A} p_i^t(a) \bar x^t_a = \beta_i^t c_i(\bar x^t) = \beta_i^t \cdot \frac{c_i(x^t)}{S} = \beta_i^t \cdot \frac{1}{n\beta_i^{t-1} S} = \frac{1}{n}\cdot\frac{\beta_i^t/\beta_i^{t-1}}{S}.
	\]
	Since $|\beta_i^t/\beta_i^{t-1} - 1| \le \eta$ for every $i \in N$, we have $1-\eta \le S \le 1+\eta$. Therefore
	\[
	\left|p_i^t \cdot \bar x^t - \frac1n\right|
	= \frac{|\beta_i^t/\beta_i^{t-1}-S|}{nS}
	\le \frac{2\eta}{n(1-\eta)}
	\le \frac{4\eta}{n}
	\le 4\eta,
	\]
	where the penultimate transition uses $\eta \le 1/2$. Hence the additive approximate spending constraint holds with error at most $4\eta$.
	
	\emph{Cost minimization:} Let $i \in N$ and $y \in \R_{\geq 0}^A$ satisfy $p_i^t \cdot y \ge p_i^t \cdot \bar x^t$. By the definition of the personalized prices,
	\[
	p_i^t \cdot y = \beta_i^t c_i(y)
	\quad\text{and}\quad
	p_i^t \cdot \bar x^t = \beta_i^t c_i(\bar x^t).
	\]
	Since $\beta_i^t > 0$, it follows that $c_i(y) \ge c_i(\bar x^t)$.
\end{proof}

\begin{proof}[Proof of \Cref{thm:frank-wolfe}]
	For each $i \in N$ let $\alpha_i = c_i(a_i)$, where $a_i \in A$ is the alternative that only imposes costs on $i$ and let $\delta = \min \{ c_i(a) : i \in N, a \in A, c_i(a) > 0 \}$ be the smallest positive cost in the instance. Note that $\delta > 0$ by the standing assumption that $0 \not\in \calC$. Then for any feasible solution to program \eqref{eq:dual}, by looking at the constraint corresponding to $a_i$, we see that $\alpha_i \beta_i \ge 1$, so $\beta_i \ge 1/\alpha_i$. Therefore the objective function value of any feasible solution of the program is upper bounded as follows:
	\[
	-\sum_{i\in N} \tfrac{1}{n} \log\beta_i \le -\sum_{i\in N} \tfrac{1}{n} \log(1/\alpha_i) = \tfrac{1}{n} \sum_{i\in N} \log\alpha_i.
	\]
	Thus, we can invoke \Cref{thm:gfw-convergence-rate} (convergence of Greedy Frank-Wolfe) with the algorithm starting from the initial point $\beta^0$ with $\beta^0_i = 1/\delta$ for all $i \in N$ (which is a feasible solution of program \eqref{eq:dual}).
	We have $\min_i a_i = \frac1n$ and $g^* - g(\beta^0)  \le \tfrac{1}{n} \sum_{i\in N} \log \frac{\alpha_i}{\delta}$.
	Given a target equilibrium approximation $0 < \epsilon \le 1$, run the algorithm with stopping tolerance $\eta = \epsilon/4$. In particular, $\eta \le 1/4$, so the conclusion of \Cref{thm:gfw-convergence-rate} holds after $O(\frac{\sum_i\log(\alpha_i/\delta)}{\eta^2}) = O(\frac{\sum_i\log(\alpha_i/\delta)}{\epsilon^2})$ iterations. At termination, the stopping criterion implies $|\beta_i^t/\beta_i^{t-1}-1| \le \eta$ for every $i \in N$. Therefore, by \Cref{lem:convergence-gives-approx-lindahl}, the dual variables encode a $4\eta = \epsilon$-approximate zero-respecting Lindahl equilibrium.
\end{proof}

\section{Discussion}\label{sec:discussion}

While we were able to provide a thorough treatment of PF and Lindahl equilibrium as applied to the public bads setting, our work leaves much to do. Arguably, the most pressing open question raised by this research is whether there exists a natural rule that combines the advantages of strictly positive PF solutions and the Flipped-MNW rule. Namely, is there a rule that works well in axes-cutting and private-induced instances, guaranteeing strong properties such as the BE core there, while also achieving some notion of fairness on general instances, such as the completion core or even just IFS?

In a similar vein, designing better fairness definitions for the public bads setting is another important future direction. Bounded-Externality and the Completion core both capture intuitive fairness properties in the settings they work well in, but neither of them can be seen as the ``ultimate'' definition of fairness across all public bads instances. Finding a clear interpretation of ``The Core'' for this setting that is not too weak in some instances, and does not only achieve fairness properties in special cases, would help shed light on the unique structure of the public bads setting, and would likely lead directly to some natural rule that could satisfy it.

\section*{AI Acknowledgment}
\Cref{lem:pf-characterization} and \Cref{ex:private-bads-pf-characterization-fails} were derived using GPT-5.2-Thinking. The model was tasked with proving the corresponding characterization for private bads due to \citet{BMSY17}. It found \Cref{ex:private-bads-pf-characterization-fails}, proving that the characterization of \citet{BMSY17} does not extend to public bads, and the necessary change (upward closure instead of Pareto optimality) required to make the characterization work for public bads, along with a complete proof of \Cref{lem:pf-characterization}.

\section*{Funding Acknowledgments}
Cookson, Ebadian, and Shah were supported by an NSERC Discovery Grant and an NSERC-CSE Research Communities Grant. Researchers funded through the NSERC-CSE Research Communities Grants do not represent the Communications Security Establishment Canada or the Government of Canada. Any research, opinions or positions they produce as part of this initiative do not represent the official views of the Government of Canada. Peters was funded in part by the Agence Nationale de la Recherche as part of the France 2030 program under grant ANR-23-IACL-0008 (PR[AI]RIE-PSAI).

\bibliographystyle{ACM-Reference-Format}
\bibliography{abb, nisarg, ultimate}

\newpage
\appendix
\section*{\centering{\LARGE Appendix}}

\section{Missing Details from \Cref{sec:public-goods}}\label{app:public-goods}

\publicGoods*
\begin{proof}
    We prove the equivalence via the cycle Lindahl $\Rightarrow$ PF $\Rightarrow$ MNW $\Rightarrow$ Lindahl, and then establish the additional properties.

    \proofstep{Lindahl $\Rightarrow$ PF.}
    Let $(x,p)$ be a Lindahl equilibrium. We first argue that $v_i(x) > 0$ for every agent $i \in N$. By assumption, agent $i$ has $v_i(a) > 0$ for some alternative $a \in A$.
    If $p_i(a) = 0$ for that $a$, then agent $i$ could purchase an unbounded amount of $a$ within budget, contradicting the existence of a finite optimum; hence $p_i(a) > 0$.
    The vector $y \in  \R_{\geq 0}^A$ such that $y_a = \frac{1}{n \cdot p_i(a)}$ and $y_{a'} = 0$ for all $a' \neq a$ then satisfies agent $i$'s spending constraint (with $p_i \cdot y = \nicefrac{1}{n}$)
    and yields utility $v_i(y) = \frac{v_i(a)}{n \cdot p_i(a)} > 0$, so utility maximization gives $v_i(x) \ge v_i(y) > 0$.

    Now, by utility maximization, the KKT conditions for agent $i$'s problem $\max_{y \ge 0 : p_i \cdot y \le 1/n} v_i(y)$ yield a multiplier $\lambda_i \ge 0$ such that
    \[
        v_i(a) \le \lambda_i \, p_i(a) \quad \text{for all } a \in A,
    \]
    with equality whenever $x(a) > 0$. Multiplying by $x(a)$ and summing over $A$ gives
    \[
        v_i(x) = \lambda_i \cdot (p_i \cdot x) = \frac{\lambda_i}{n},
    \]
    where the last step uses budget exhaustion ($p_i \cdot x = \nicefrac{1}{n}$, which follows from utility maximization together with $v_i(x) > 0$). Hence $\lambda_i = n \cdot v_i(x)$. For any $y \in \Delta(A)$, summing $v_i(y) \le \lambda_i(p_i \cdot y)$ over all agents yields
    \[
        \sum_{i \in N} \frac{v_i(y)}{v_i(x)}
        \le n \sum_{i \in N} p_i \cdot y
        \le n \sum_{a \in A} y(a)
        = n,
    \]
    where we used profit maximization ($\sum_{i \in N} p_i(a) \le 1$ for all $a$). This is precisely PF.

    \proofstep{PF $\Rightarrow$ MNW.}
    Let $x$ be a PF lottery. We first show $v_i(x) > 0$ for all $i \in N$. Indeed, for each agent $i$, there exists some $a \in A$ with $v_i(a) > 0$. If $v_i(x) = 0$, then $\nicefrac{v_i(a)}{v_i(x)} = +\infty$ by our convention, so $\frac{1}{n}\sum_{j \in N}\frac{v_j(a)}{v_j(x)} = +\infty > 1$, violating PF. Thus $v_i(x) > 0$ for all $i$.

    Consider the concave function $F(u) \triangleq \sum_{i \in N} \log u_i$ on $\R_{> 0}^N$. The PF condition states that for all $y \in \Delta(A)$,
    \[
        \nabla F(v(x)) \cdot v(y) = \sum_{i \in N} \frac{v_i(y)}{v_i(x)} \le n = \sum_{i \in N} \frac{v_i(x)}{v_i(x)} = \nabla F(v(x)) \cdot v(x).
    \]
    Since $\calU = \set{v(y) : y \in \Delta(A)}$, this means $\nabla F(v(x)) \cdot (u - v(x)) \le 0$ for all $u \in \calU$. Because $F$ is concave and $\calU$ is convex, this first-order condition implies that $v(x)$ is a global maximizer of $F$ over $\calU$. Hence $x$ maximizes $\sum_{i \in N}\log v_i(\cdot)$, and equivalently $\prod_{i \in N} v_i(\cdot)$, over $\Delta(A)$.

    \proofstep{MNW $\Rightarrow$ Lindahl.}
    Let $x$ be an MNW lottery. We first show $v_i(x) > 0$ for all $i \in N$. Since each agent values at least one alternative positively, any full-support lottery has strictly positive Nash welfare; hence any MNW lottery must also have $\NW(x) > 0$, which forces $v_i(x) > 0$ for all $i$.

    Now, $x$ maximizes the concave function $F(y) \triangleq \sum_{i \in N} \log v_i(y)$ over the simplex $\Delta(A)$. The KKT conditions yield a scalar $\mu \in \R$ such that
    \begin{align*}
        \sum_{i \in N} \frac{v_i(a)}{v_i(x)} &= \mu \quad \text{for all } a \in \supp(x), \\
        \sum_{i \in N} \frac{v_i(a)}{v_i(x)} &\le \mu \quad \text{for all } a \notin \supp(x).
    \end{align*}
    Multiplying the equalities by $x(a)$ and summing over $\supp(x)$ gives $\sum_{i \in N} 1 = \mu$, so $\mu = n$. Define personalized prices
    \[
        p_i(a) \triangleq \frac{1}{n} \cdot \frac{v_i(a)}{v_i(x)} \quad \text{for all } i \in N,\; a \in A.
    \]
    We verify that $(x,p)$ is a Lindahl equilibrium.

    \emph{Profit maximization:} $\sum_{i \in N} p_i(a) = \frac{1}{n}\sum_{i \in N}\frac{v_i(a)}{v_i(x)} = 1$ for $a \in \supp(x)$ and $\le 1$ for $a \notin \supp(x)$.

    \emph{Spending constraint:} $p_i \cdot x = \frac{1}{n \cdot v_i(x)}\sum_{a \in A} v_i(a) \, x(a) = \frac{v_i(x)}{n \cdot v_i(x)} = \frac{1}{n}$ for all $i$.

    \emph{Utility maximization:} The bang-per-buck ratio $\frac{v_i(a)}{p_i(a)} = n \cdot v_i(x)$ is constant across all $a \in A$ with $v_i(a) > 0$ (and alternatives with $v_i(a) = 0$ contribute zero utility regardless). Hence, for any $y \ge 0$ with $p_i \cdot y \le \nicefrac{1}{n}$,
    \[
        v_i(y) = \sum_{a : v_i(a) > 0} \frac{v_i(a)}{p_i(a)} \cdot p_i(a) \cdot y(a) \le n \cdot v_i(x) \cdot (p_i \cdot y) \le n \cdot v_i(x) \cdot \frac{1}{n} = v_i(x),
    \]
    so $x$ is optimal for agent $i$.

    \proofstep{Core membership.}
    Let $x$ be PF and suppose for contradiction that some $S \subseteq N$ and $y \in \Delta(A)$ satisfy $\frac{|S|}{n} \cdot v_i(y) \ge v_i(x)$ for all $i \in S$, with at least one strict inequality. Then
    \[
        \frac{1}{n} \sum_{i \in N} \frac{v_i(y)}{v_i(x)}
        \ge \frac{1}{n}\sum_{i \in S} \frac{v_i(y)}{v_i(x)}
        > \frac{1}{n} \cdot |S| \cdot \frac{n}{|S|} = 1,
    \]
    contradicting PF. Hence all PF (equivalently, Lindahl and MNW) lotteries lie in the core.

    \proofstep{Utility uniqueness and strict positivity.}
    We showed above that every Lindahl/PF/MNW lottery $x$ satisfies $v_i(x) > 0$ for all $i \in N$. Since $F(u) = \sum_{i \in N}\log u_i$ is strictly concave on $\R_{>0}^N$ and $\calU$ is convex, the maximizer of $F$ over $\calU$ is unique. Hence all such lotteries induce the same utility vector.
\end{proof}

\section{Missing Details from \Cref{sec:public-bads}}\label{app:public-bads}

\begin{proposition}\label{prop:private-bads-core}
    Every CEEI allocation of a private bads instance with linear costs lies in the core.
\end{proposition}
\begin{proof}
    Let $x$ be a CEEI allocation. 
    
    By \citet{BMSY17}, $x$ is PF and satisfies $c_i(x)>0$ for all $i\in N$. 
    
    Suppose for contradiction that some nonempty $S\subseteq N$ and allocation $y$ of the bads to the agents in $S$ satisfy
    \[
        \frac{|S|}{n}\,c_i(y)\le c_i(x)
    \]
    for every $i\in S$, with at least one strict inequality. 
    
    Since $y$ allocates no bads to agents outside $S$ and $c_i(x) > 0$ for all $i \in N$, we obtain
    \[
        \frac{1}{n}\sum_{i\in N}\frac{c_i(y)}{c_i(x)}
        =\frac{1}{n}\sum_{i\in S}\frac{c_i(y)}{c_i(x)}
        <\frac{1}{n}\cdot |S|\cdot\frac{n}{|S|}
        =1,
    \]
    where the strict inequality follows from the assumed core-blocking deviation. This contradicts the proportional fairness of $x$.
\end{proof}

\section{Missing Details from \Cref{sec:bads-lindahl}}\label{app:bads-lindahl}

\begin{theorem}\label{thm:be-core-private-core}
    For any private bads instance $(N,B,c)$, if $x$ is an allocation for that instance and $y$ is an allocation over the private-induced public bads instance $\calR(N,B,c)$ induced by $x$, then $y$ being in the bounded-externality core implies that $x$ is in the core for private bads.
\end{theorem}
\begin{proof}
    For a contradiction, assume this is false, and $y$ is in the BE core while $x$ is not in the core for private bads. Thus, there exists some nonempty $S \subseteq N$, and some alternate private bads allocation $x'$ which only allocates the bads to agents in $S$, such that for each $i \in S$,
    \[\frac{|S|}{n}c_i(x') \leq c_i(x),\]
    and this inequality is strict for at least one agent in $S$.

    Let $y'$ be a public allocation for $\calR(N,B,c)$ that is induced by $x'$. Consider the potential BE core deviation of $S$ with $y'$, and $\alpha = 0$.
    Since in $x'$, all the private bads were allocated among the agents in $S$, we must have that $c_i(x') = c_i(y') = 0$ for all $i \in N \setminus S$, and thus $0 = c_i(y') \leq 0 \cdot c_i(y)$. So the BE core violation condition on $N \setminus S$ is met.
    Then, for each $i \in S$, we have $\frac{n - |N \setminus S|\alpha}{|S|}c_i(x) = \frac{n}{|S|}c_i(x)$. Therefore, for each of these agents, the inequality $c_i(y') \leq \frac{n}{|S|}c_i(x)$ (with strictness for some agent) is implied by the private bads core violation. This gives a contradiction.
\end{proof}

\begin{lemma}\label{lem:be-core-positive}
    If $x$ is an allocation in the bounded-externality core, then $c_i(x) > 0$ for all $i \in N$.
\end{lemma}
\begin{proof}
    Suppose for contradiction that $x$ is in the BE core but $c_j(x) = 0$ for some $j \in N$. Let $Z = \{i \in N : c_i(x) = 0\}$ be the (nonempty) set of agents with zero cost under $x$. By the standing assumption $0\notin\calC$, we have $c(x)\neq 0$, so $S := N \setminus Z \neq \emptyset$.

    We show that $(S, x, 0)$ is a BE core deviation of $x$, with deviating coalition $S$, alternative lottery $y = x$, and scale factor $\alpha = 0$:
    \begin{enumerate}
        \item For each $i \in N \setminus S = Z$: $c_i(y) = c_i(x) = 0 = \alpha \cdot c_i(x)$, so condition~(\ref{item:be-core-outsiders}) holds.
        \item For each $i \in S$: by definition of $S = N \setminus Z$ we have $c_i(x) > 0$, and since $|S| = n - |Z| < n$ (as $Z \neq \emptyset$):
        \[
            c_i(y) = c_i(x) < \frac{n}{|S|} \cdot c_i(x) = \frac{n - |N \setminus S| \cdot 0}{|S|} \cdot c_i(x),
        \]
        so condition~(\ref{item:be-core-deviators}) holds with strict inequality for every $i \in S$.
    \end{enumerate}
    This contradicts $x$ being in the BE core.
\end{proof}

\section{Missing Details from \Cref{sec:impossibilities}}\label{app:impossibilities}

\begin{theorem}\label{thm:PF-meets-axioms}
    The rule that returns all strictly positive PF solutions satisfies the axioms of cost consistency, weak symmetry, scale-freeness, and lower contraction consistency.
\end{theorem}
\begin{proof}
    \textbf{Cost consistency:} Clearly, since the definition of strictly positive PF depends only on the cost vector of lotteries, if two lotteries $x$ and $x'$ induce the same cost vector, then $x$ will be strictly positive PF if and only if $x'$ is.

    \textbf{Weak symmetry:} Fix the identity instance for some number of agents $n$. The allocation $x$ that uniformly mixes over all options clearly gives each agent positive cost. To see that it is a PF solution, note that $c_i(x) = \nicefrac{1}{n}$ for all $i \in N$. Thus, for every alternative lottery $y$, we have
\[\sum_{i \in N}\frac{c_i(y)}{c_i(x)} = n\sum_{i \in N}c_i(y) = n,\]
    where the final equality follows from the fact that each alternative in this instance gives a cost of $1$ to a single agent and $0$ to everyone else, so $\sum_{a \in A}y(a)\sum_{i \in N}c_i(a) = \sum_{a \in A}y(a) = 1$ for any lottery $y$.

    Next, for contradiction, assume that there is some lottery $y \neq x$ that is strictly positive PF. Since $y$ does not uniformly mix between all options, there must be some bad $a \in A$ with $y(a) > \nicefrac{1}{n}$. Assume that $i^*$ is the agent such that $c_{i^*}(a) = 1$. Hence
\[\sum_{i \in N}\frac{c_i(a)}{c_i(y)} = \frac{c_{i^*}(a)}{c_{i^*}(y)} < \frac{1}{\nicefrac{1}{n}} = n,\]
    meaning that $y$ is not PF, a contradiction. Note that the first equality crucially uses the fact that $c_i(a) = 0$ for all $i \neq i^*$, and $c_i(y) > 0$ for all $i \in N$.

    \textbf{Scale-freeness:} Consider any two instances with cost functions $c,c'$ that are identical except for the cost function of some agent $i^*$, whose costs are all scaled by a factor of $k>0$ ($c'_{i^*}(a) = kc_{i^*}(a)$ for all $a \in A$).

    Clearly, for any two allocations $x, y$, we have
\[\sum_{i \in N}\frac{c'_{i}(y)}{c'_{i}(x)} = \sum_{i \in N \setminus \set{i^*}}\frac{c_i(y)}{c_i(x)} + \frac{kc_{i^*}(y)}{kc_{i^*}(x)} = \sum_{i \in N}\frac{c_i(y)}{c_i(x)}.\]
    Positive scaling also preserves whether every agent has strictly positive cost. Thus, $x$ is strictly positive PF under $c$ if and only if it is strictly positive PF under $c'$.

    \textbf{Lower contraction consistency:} Assume we have two instances $I_1 = (N,A,c)$ and $I_2 = (N,A',c')$ such that $\calC_{\geq}(I_1) \subseteq \calC_{\geq}(I_2)$. For contradiction, assume there is a strictly positive PF allocation $x'$ for the instance $I_2$, and a lottery $x$ for the instance $I_1$ such that $c_i(x) = c'_i(x')$ for all $i \in N$, but $x$ is not a strictly positive PF lottery for the instance $I_1$. Then there must be some lottery $y$ for the instance $I_1$ such that
\[\sum_{i \in N}\frac{c_i(y)}{c_i(x)} < n.\]
    However, since $\calC_{\geq}(I_1) \subseteq \calC_{\geq}(I_2)$, there must also be a lottery $y'$ over the instance $I_2$ such that $c'_i(y') \leq c_i(y)$ for all $i \in N$. But this would give us
\[\sum_{i \in N}\frac{c'_i(y')}{c'_i(x')} < n,\]
    contradicting the proportional fairness of $x'$ for instance $I_2$.
\end{proof}

\begin{theorem}
    Any public bads allocation rule that satisfies the axioms of cost consistency, weak symmetry, scale-freeness, and lower contraction consistency must return a superset of the strictly positive proportionally fair lotteries for every instance.
\end{theorem}
\begin{proof}
    Suppose for contradiction that some rule $\phi$ satisfies all four axioms, yet for some instance $I = (N,B,c)$ there is a strictly positive PF lottery $x$ with $x \notin \phi(I)$.

    \textbf{Step 1: Normalize.}
    Since $x$ is strictly positive PF, $c_i(x) > 0$ for all $i \in N$. Define a scaled instance $I' = (N,B,c')$ by $c'_i(b) = c_i(b)/(n \cdot c_i(x))$ for each $i \in N$ and bad $b$. Applying scale-freeness once per agent gives $\phi(I) = \phi(I')$, so $x \notin \phi(I')$. The scaling is chosen so that
    \[
        c'_i(x) = \frac{c_i(x)}{n \cdot c_i(x)} = \frac{1}{n} \quad \text{for all } i \in N.
    \]
    Since strictly positive PF is scale-free (\Cref{thm:PF-meets-axioms}), $x$ remains strictly positive PF in $I'$. Hence for every lottery $y$ in $I'$,
    \[
        n \leq \sum_{i \in N}\frac{c'_i(y)}{c'_i(x)} = n\sum_{i \in N}c'_i(y),
    \]
    so every cost vector achievable in $I'$ satisfies $\sum_{i \in N} c'_i(y) \geq 1$.

    \textbf{Step 2: Compare with the identity instance.}
    Let $\calI$ be the identity instance for $n$ agents, and let $d$ denote its cost function. Its $n$ alternatives have cost vectors $e_1,\ldots,e_n$, giving cost set $\calC(\calI) = \conv\{e_1,\ldots,e_n\}$, where $e_i$ is the $i$-th standard basis vector. Thus
    \[
        \calC_{\geq}(\calI) = \bigl\{q \in \R^n_{\geq 0} : \textstyle\sum_{i \in N} q_i \geq 1\bigr\}.
    \]
    Step~1 shows every cost vector of $I'$ lies in $\calC_{\geq}(\calI)$, so $\calC_{\geq}(I') \subseteq \calC_{\geq}(\calI)$.

    \textbf{Step 3: Apply lower contraction consistency.}
    By weak symmetry, $\phi(\calI)$ contains only the uniform lottery $u$ (which places weight $\nicefrac{1}{n}$ on each bad), giving $d_i(u) = \nicefrac{1}{n}$ for all $i$. Since $x$ achieves $c'_i(x) = \nicefrac{1}{n}$ in $I'$, the cost vector $(\nicefrac{1}{n},\ldots,\nicefrac{1}{n})$ lies in $\calC(I')$.

    We apply LCC with $I'$ as the smaller instance and $\calI$ as the larger: the hypothesis $\calC_{\geq}(I') \subseteq \calC_{\geq}(\calI)$ holds, and $\phi(\calI)$ selects $u$ with $d(u) \in \calC(I')$. LCC therefore gives
    \[
        d(\phi(\calI)) \cap \calC(I') \;\subseteq\; c'(\phi(I')).
    \]
    We have $d(u) = (\nicefrac{1}{n},\ldots,\nicefrac{1}{n}) \in d(\phi(\calI)) \cap \calC(I')$, so $(\nicefrac{1}{n},\ldots,\nicefrac{1}{n}) \in c'(\phi(I'))$: some lottery $x' \in \phi(I')$ satisfies $c'_i(x') = \nicefrac{1}{n}$ for all $i$.

    \textbf{Step 4: Conclude.}
    Since $c'_i(x') = c'_i(x) = \nicefrac{1}{n}$ for all $i$, cost consistency gives $x \in \phi(I')$. But $\phi(I') = \phi(I)$, so $x \in \phi(I)$, contradicting our assumption.
\end{proof}

The following result shows that the completion core does not imply the core for
private bads (\Cref{def:core-private-bads}) on private-induced instances, and
in particular that Flipped-MNW can violate the latter.

\begin{proposition}\label{prop:completion-core-not-private-core}
    There exists a private bads instance $(N,B,c)$ for which an allocation
    returned by Flipped-MNW on the induced public bads instance $\calR(N,B,c)$
    is not in the core for private bads and does not achieve strong individual fair share.
\end{proposition}

\begin{proof}
    Consider $n = 3$ agents and $m = 2$ bads with costs:
    \begin{center}
    \begin{tabular}{ccc}
    \toprule
     & Bad $b_1$ & Bad $b_2$ \\
    \midrule
    Agent 1 & 1 & 5 \\
    Agent 2 & 1 & 6 \\
    Agent 3 & 6 & 6 \\
    \bottomrule
    \end{tabular}
    \end{center}

    Let $A$ denote the alternative set of the induced instance $\calR(N,B,c)$, which consists of the $3^2 = 9$ complete assignments. To apply
    Flipped-MNW, we compute $c_i^{\max} = \max_{a \in A} c_i(a)$ for each
    agent, where the maximum ranges over all complete assignments. Since
    assigning both bads to agent $i$ maximizes their cost, we get
    $c_1^{\max} = 6$, $c_2^{\max} = 7$, and $c_3^{\max} = 12$. The
    Flipped-MNW rule then maximizes Nash welfare in the dual goods instance
    with valuations $v_i(a) = c_i^{\max} - c_i(a)$.

    We claim that the MNW lottery is supported on the two alternatives
    $a^* \colon b_1 \to 2,\, b_2 \to 1$ (with $v(a^*) = (1, 6, 12)$) and
    $a^{**} \colon b_1 \to 2,\, b_2 \to 3$ (with $v(a^{**}) = (6, 6, 6)$).
    Writing $t$ for the probability on $a^*$, note that $v_2 = 6$ is constant,
    so the Nash welfare is proportional to
    $(6 - 5t)(6 + 6t)$,
    which is maximized at $t = \nicefrac{1}{10}$. One can verify that no
    alternative outside $\set{a^*, a^{**}}$ violates the PF condition for the
    dual goods instance, confirming this is the MNW lottery.\footnote{The maximum PF ratio among the seven remaining alternatives is
        $\nicefrac{1}{3}\sum_{i} v_i(a)/v_i(x) \approx 0.995$, attained by
        the assignments $b_1 \to 1, b_2 \to 1$ and $b_1 \to 1, b_2 \to 3$.}

    Translating back to costs, the Flipped-MNW allocation $x$ induces
    \[
        (c_1(x),\, c_2(x),\, c_3(x))
        = \tfrac{1}{10}(5, 1, 0) + \tfrac{9}{10}(0, 1, 6)
        = \bigl(\tfrac{1}{2},\; 1,\; \tfrac{27}{5}\bigr).
    \]
    This allocation is in the completion core by the general guarantee that
    every Flipped-MNW allocation is.

    However, the coalition $S = \set{2, 3}$ blocks $x$ under the core for
    private bads (\Cref{def:core-private-bads}). Consider the deviation $y$
    assigning $b_1$ entirely to agent~2 and $b_2$ entirely to agent~3, so that
    $c_2(y) = 1$ and $c_3(y) = 6$, with $y_{j}(b) = 0$ for $j \notin S$.
    Then
    \[
        \frac{|S|}{n} \cdot c_2(y) = \frac{2}{3} < 1 = c_2(x),
        \qquad
        \frac{|S|}{n} \cdot c_3(y) = 4 < \frac{27}{5} = c_3(x),
    \]
    with both inequalities strict, so $S$ blocks.

    Also, note that sIFS would require that agent $3$ receive a cost no more than $\nicefrac{12}{3} = 4$, but they receive a cost of $\nicefrac{27}{5}$.
\end{proof}

\begin{definition}[Participation]
    A rule $\phi$ satisfies participation if, for any two instances $I = (N,A,c)$ and $I' = (N \cup \{i\},A,c')$ such that $c'_j = c_j$ for every $j \in N$, we have that for every $x \in \phi(I')$, there exists a $y \in \phi(I)$ such that $c'_i(x) \leq c'_i(y)$. In other words, for every selection of $\phi(I')$, $i$ should weakly improve over some selection of $\phi(I)$ by joining the instance.
\end{definition}

We can show that the axioms that characterize strictly positive PF are also incompatible with this intuitive definition.

\begin{theorem}\label{thm:impossibility-participation}
    No rule which satisfies weak symmetry, scale-freeness, and lower contraction consistency can also satisfy participation.
\end{theorem}
\begin{proof}
    Consider some rule $\phi$ that satisfies the three axioms. First, consider the identity instance $\calI_3$ with $3$ agents. Since $\phi$ is weakly symmetric, it must return the lottery that uniformly mixes between all $3$ bads on this instance.

    \begin{table}[H]
    \centering
    \begin{tabular}{cccc}
    \toprule
              & Bad $1$ & Bad $2$ & Bad $3$ \\
    \midrule
    Agent $1$ & $1$     & $0$     & $0$ \\
    Agent $2$ & $0$     & $1$     & $0$ \\
    Agent $3$ & $0$     & $0$     & $1$ \\
    \bottomrule
    \end{tabular}
    \end{table}

    Next, consider the following instance $I$, constructed by taking the identity instance and scaling the costs of all agents except agent $3$ by a factor of $2$.

    \begin{table}[H]
    \centering
    \begin{tabular}{cccc}
    \toprule
              & Bad $1$ & Bad $2$ & Bad $3$ \\
    \midrule
    Agent $1$ & $2$     & $0$     & $0$ \\
    Agent $2$ & $0$     & $2$     & $0$ \\
    Agent $3$ & $0$     & $0$     & $1$ \\
    \bottomrule
    \end{tabular}
    \end{table}

    Due to scale-freeness, $\phi$ will still only return the lottery $x$ that uniformly mixes between all three bads. Under these scaled costs, this lottery will give agents $1$ and $2$ a cost of $\nicefrac{2}{3}$ and agent $3$ a cost of $\nicefrac{1}{3}$.

    Then, consider the instance $I'$ with $3$ agents and $2$ bads.

    \begin{table}[H]
    \centering
    \begin{tabular}{ccc}
    \toprule
              & Bad $1$ & Bad $2$ \\
    \midrule
    Agent $1$ & $1$     & $0$     \\
    Agent $2$ & $1$     & $0$    \\
    Agent $3$ & $0$     & $1$      \\
    \bottomrule
    \end{tabular}
    \end{table}

    We can show that $\calC_{\geq}(I') \subseteq \calC_{\geq}(I)$. To see this, consider an arbitrary lottery over instance $I'$ that places weight $\alpha$ on Bad $1$ and weight $1 - \alpha$ on Bad $2$. This lottery will give agents $1$ and $2$ a cost of $\alpha$, and agent $3$ a cost of $1-\alpha$.

    Now consider the lottery over $I$ that places weight $\nicefrac{\alpha}{2}$ on each of Bads $1$ and $2$, and weight $1 - \alpha$ on Bad $3$. This will also give all agents the same costs as the first lottery over $I'$. This shows that $\calC(I') \subseteq \calC(I)$, and hence $\calC_{\geq}(I') \subseteq \calC_{\geq}(I)$ as well.

    With this, consider the lottery $x'$ on $I'$ that places $\nicefrac{2}{3}$ weight on Bad $1$ and $\nicefrac{1}{3}$ weight on Bad $2$. This will give the first $2$ agents a cost of $\nicefrac{2}{3}$ and agent $3$ a cost of $\nicefrac{1}{3}$, the exact same cost vector as the lottery $x$, which we know that $\phi$ returns for instance $I$. Since $\phi$ is lower contraction consistent and $\phi(I)$ returns a lottery whose cost vector is feasible in $I'$, $\phi(I')$ must return some lottery with this cost vector. The only lottery on $I'$ with this cost vector is $x'$, so $x' \in \phi(I')$.

    Finally, consider the identity instance $\calI_2$ on $2$ agents.

    \begin{table}[H]
    \centering
    \begin{tabular}{ccc}
    \toprule
              & Bad $1$ & Bad $2$ \\
    \midrule
    Agent $1$ & $1$     & $0$ \\
    Agent $3$ & $0$     & $1$ \\
    \bottomrule
    \end{tabular}
    \end{table}

    Note that this instance can be thought of as $I'$ without the participation of Agent $2$. By weak symmetry, the unique lottery returned by $\phi(\calI_2)$ must place weight $\nicefrac{1}{2}$ on each bad. However, this violates participation, as in instance $I'$, $c_2(x') = \nicefrac{2}{3}$, which is strictly worse than the cost of $\nicefrac{1}{2}$ for agent $2$ induced by the unique lottery returned by $\phi(\calI_2)$.
\end{proof}

However, Flipped-MNW does satisfy participation.

\begin{theorem}
    Flipped-MNW satisfies participation.
\end{theorem}
\begin{proof}
    This proof follows quite easily from the fact that MNW is known to satisfy a similar participation-style axiom in the case of public goods \citep{aziz2019fair, brandl2022funding, airiau2023portioning}. For completeness, we will prove it in full here.

    Consider two public bads instances $I=(N,A,c)$ and $I'=(N\cup\{i^*\},A,c')$, where $I'$ is constructed by adding agent $i^*$ to $I$.
    For each $i \in N \cup \{i^*\}$, let $v_i$ be the valuation function of $i$ in the dual goods problem created by the Flipped-MNW rule.
    We know that running Flipped-MNW on $I$ will return all lotteries $x \in \argmax_{x' \in \Delta(A)}\prod_{i \in N}v_i(x')$, and running Flipped-MNW on $I'$ will return all lotteries $y \in \argmax_{x' \in \Delta(A)}v_{i^*}(x')\cdot\prod_{i \in N}v_i(x')$.
    Let $\phi$ represent the Flipped-MNW rule. For contradiction, assume that $\phi$ violates participation on these instances, meaning that there exists some $y \in \phi(I')$ such that $c'_{i^*}(y) > c'_{i^*}(x)$ for all lotteries $x \in \phi(I)$.

    Fix some $x \in \phi(I)$. Then we must have
    
    \[v_{i^*}(y)\cdot\prod_{i \in N}v_i(x) \geq v_{i^*}(y)\cdot\prod_{i \in N}v_i(y) \geq v_{i^*}(x)\cdot\prod_{i \in N}v_i(x),\]
    where the first inequality follows from the fact that $x$ maximizes the Nash welfare among agents in $N$, and the second follows from the fact that $y$ maximizes the Nash welfare among agents in $N \cup \{i^*\}$.

    We know from our non-degeneracy assumptions on the original public bads instance that $v_i(x) > 0$ for all $i \in N$, so we can conclude that $v_{i^*}(y) \geq v_{i^*}(x)$. Substituting the definition of $v_{i^*}$ gives

    \[\max_{a \in A}c'_{i^*}(a) - c'_{i^*}(y) \geq \max_{a \in A}c'_{i^*}(a) - c'_{i^*}(x) \Longrightarrow c'_{i^*}(x) \geq c'_{i^*}(y),\]
    the desired contradiction.
\end{proof}

\end{document}